\documentclass[journal]{IEEEtran}

\usepackage{multicol}
\usepackage{multirow}
\usepackage{longtable}
\usepackage{subcaption}

\usepackage{balance}
\usepackage{cite}

\usepackage{algorithm}
\usepackage{algpseudocode}

\ifCLASSINFOpdf
   \usepackage[pdftex]{graphicx}
   \DeclareGraphicsExtensions{.pdf,.jpeg,.png}
\else
 
\fi

\usepackage{mathrsfs,amsmath}
\usepackage{amssymb}

\usepackage{amsthm,hyperref,enumerate}
\newtheorem{theorem}{Theorem}

\newtheorem{problem}{Problem}

\theoremstyle{definition}
\newtheorem{defn}{Definition}

\theoremstyle{remark}
\newtheorem{remark}{Remark}
\newtheorem{lemma}{Lemma}

\theoremstyle{plain}
 \newtheorem{assumption}{Assumption}

\newcommand{\tabincell}[2]{\begin{tabular}{@{}#1@{}}#2\end{tabular}}

\usepackage{ragged2e}

\usepackage[amssymb]{SIunits}
\hypersetup{
	colorlinks=true,
	linkcolor=blue,
	filecolor=blue,      
	urlcolor=blue,
}

\begin{document}

\title{Angle-Displacement Rigidity Theory with Application to  Distributed Network Localization}

\author{Xu Fang, Xiaolei Li, Lihua Xie,~\IEEEmembership{Fellow,~IEEE} 
\thanks{This work is partially supported by ST Engineering-NTU Corporate Lab under the NRF Corporate Lab @ University Scheme and National Natural Science Foundation of China (NSFC) under Grant 61720106011 and 61903319. (Corresponding author: Lihua Xie.)}
\thanks{The authors are with the School of Electrical and Electronic Engineering, Nanyang Technological University, Singapore. (E-mail: fa0001xu@e.ntu.edu.sg; xiaolei@ntu.edu.sg;
elhxie@ntu.edu.sg).}
}

\maketitle

\begin{abstract}
This paper investigates the localization problem of a network in 2-D and 3-D spaces given the positions of anchor nodes in a global frame and inter-node relative measurements in local coordinate frames. 
It is assumed that the local frames of different nodes have different unknown orientations.
First, an angle-displacement rigidity theory is developed, which can be used to localize all the free nodes by the known positions of the anchor nodes and local relative measurements (local relative position, distance, local relative bearing, angle,  or ratio-of-distance measurements). Then,  necessary and sufficient conditions for network localizability are given. Finally, a distributed network localization protocol is proposed, which can globally estimate the locations of all the free nodes of
a network if the network is infinitesimally angle-displacement rigid. The proposed method 
unifies 
local-relative-position-based, distance-based, local-relative-bearing-based, angle-based, and ratio-of-distance-based distributed network localization approaches. The novelty of this work is that the proposed method can be applied in both generic and non-generic configurations with an unknown global coordinate frame in both 2-D and 3-D spaces.
\end{abstract}


\IEEEpeerreviewmaketitle

\section{Introduction}
\IEEEPARstart{L}{ocalization} is a prerequisite for large range of applications in target searching, obstacle avoidance, and formation control \cite{ zhao2018affine, lin2014fdistributed}. Network localization studies how to localize a network given the positions of anchor nodes and inter-node relative measurements. Compared with centralized network localization, distributed network localization can 
reduce the network resource consumption.  The existing distributed network localization algorithms, based on different kinds of measurements, can be divided into four categories: relative-position-based \cite{barooah2007estimation, lin2015distributed}, distance-based \cite{ diao2014barycentric, han2017barycentric}, bearing-based \cite{zhao2016localizability, li2019globally}, and angle-based \cite{jing2019angle1}.

There are two fundamental problems in network localization: \textit{network localizability} and  \textit{distributed localization protocols}.
The necessary and sufficient conditions for localizability of a network have been studied by many researchers. For the relative-position-based localization in 2-D space \cite{lin2015distributed}, a network is localizable  if and only if every free node is 2-reachable from the set of anchor nodes. Similar results have been shown in the distance-based localization \cite{diao2014barycentric}, where a network expressed by the barycentric coordinate is localizable if and only if every free node has at least three disjoint paths in 2-D space from the set of anchor nodes.  The localizability of bearing-based localization depends on the bearing rigidity theory  \cite{zhao2016localizability}. 
A bearing-based network is localizable in arbitrary dimensional spaces if and only if every infinitesimal bearing motion involves at least one anchor node. For the angle-based localization, a network is angle localizable in 2-D space if and only if it is angle fixable and anchors are not colinear \cite{jing2019angle1}. It should be noted that different localizability conditions are proposed for different relative measurements. A general condition for 2-D and 3-D network localizability for different local relative measurements is still lacking.

If a network is localizable, the corresponding relative-position-based, distance-based, bearing-based, or angle-based distributed protocols can globally estimate a network. However, there are limitations in the existing distributed protocols. 
The disadvantage of  relative-position-based and bearing-based distributed localization is the requirement of the global coordinate frame.
To address this issue, the first way to recover the true configuration is to construct a similar configuration \cite{lin2015distributed,lin2016distributed}, which is only applicable to 2-D space. The second way is to
align the orientations of each local frame with the global frame by using local relative bearing measurements \cite{eren2007using,zhong2014cooperative, van2018distributed}, 
but the orientation 
alignment errors along the sequence of propagation may be accumulated. The third way is to estimate the orientation of each local frame, but it requires additional relative orientation measurements \cite{li2019globally}.

Compared with relative-position-based and bearing-based distributed localization, the distance-based and angle-based distributed localization do not need the global coordinate frame. But the distance-based distributed network localization is only applicable to a generic configuration \cite{diao2014barycentric, han2017barycentric}, e.g., any three nodes are not on the same line in 2-D space and any four nodes are not on the same plane in 3-D space. For the angle-based network localization \cite{biswas2005semidefinite, bruck2009localization, jing2019angle1}, to the best of our knowledge, 
there is no such known result for the localizability and distributed protocols of angle-based networks in 3-D space. The works in \cite{biswas2005semidefinite, bruck2009localization, jing2019angle1} are applied to a 2-D angle-based network localization problem.
In addition,  
the unknown free nodes in \cite{jing2019angle1}
are localized one by one, i.e., the distributed algorithm works in a step-by-step iterative way.
{Moreover, for the recently developed ratio-of-distance rigidity theory \cite{cao2019ratio},
there is no 
corresponding network localization method.}

\begin{table*}[t]
    \begin{center}
    \caption{Comparison with existing distributed network localization methods.}
        \begin{tabular}{c|c c c c}
        \hline
        Approaches &References  & Measurements  & Limitations & Advantages of our method
      \\\hline
        
          Angle-based & \cite{jing2019angle1} & Angle
         & 2-D space & 3-D space \\  
         \hline
          Bearing-based & \cite{zhao2016localizability} & \hspace*{-0.15cm} Relative bearing
         &  Known global coordinate frame
        & \tabincell{c}{Unknown global coordinate frame}
        \\
          Bearing-based & \cite{lin2016distributed} & \hspace*{-0.15cm}  Local relative bearing &
         2-D space  & 3-D space
        \\
        \hline
         Distance-based &  \cite{diao2014barycentric, han2017barycentric} &Relative distance & Generic configuration  & Non-generic configuration    \\
         \hline
         Relative-position-based & \cite{barooah2007estimation, stacey2017role} & Relative position
        &  Known global coordinate frame
        & \tabincell{c}{Unknown global coordinate frame}
         \\
         Relative-position-based & \cite{lin2015distributed} & 
       Local relative position
         & 2-D space, generic configuration &   3-D space, non-generic configuration  \\
        \hline
        Ratio-of-distance-based & $-$ & Ratio-of-distance & $-$ & Solved in our work \\
        \hline
        \end{tabular}
\label{tab:accuracy-orb}
\\~\\
\begin{flushleft}
$-$: To the best of our knowledge, 
there exists no result for the localizability and distributed protocols of  ratio-of-distance-based networks. 

Compared with the distance-based localization protocol in \cite{aspnes2006theory}, the proposed method requires fewer communication channels.
\end{flushleft}
\end{center}
\end{table*}

In this paper, we aim to propose a unified approach
which not only can be applied to distributed network localization with local relative measurements (local relative position, distance, local relative bearing, angle and ratio-of-distance measurements), but also has advantages over the existing distributed network localization methods. The novelty of this work is that the proposed method can be applied in both generic and non-generic configurations with an unknown global coordinate frame in both 2-D and 3-D spaces. The comparison with existing distributed network localization methods is given in Table \ref{tab:accuracy-orb}.
The main contributions of this work are summarized below.

\begin{enumerate}
\item An angle-displacement rigidity theory is proposed to determine whether a set of angle and displacement constraints can uniquely characterize a network in 3-D space up to directly similar transformation, i.e., translation, rotation, and scaling. 
\item Based on the angle-displacement rigidity theory and an angle-displacement information matrix, we provide necessary and sufficient conditions for network localizability. 
\item A distributed network localization protocol is proposed, which can globally estimate the locations of all free nodes of
a network if the network is infinitesimally angle-displacement rigid.
\end{enumerate}

The remainder of this paper is organized as follows. Section \ref{pret} provides preliminaries and introduces {how to establish the direct similarity of two
networks by 
using angle and displacement constraints.} An angle-displacement rigidity theory is presented in Section \ref{hyper}. Section \ref{fomu} formulates the problem of displacement-constraint-based network localization, where the network localizability is analyzed. A distributed network localization protocol is presented in Section \ref{distri}. Section \ref{simulation} provides the details of simulation. Section \ref{coc} ends this paper with conclusions and recommendations for future work.

\section{Preliminaries}\label{pret}

\subsection{Notations}

The set of real numbers and real matrices are denoted by $\mathbb{R}$ and $\mathbb{R}^{m \times n}$, respectively. We use $\text{Null}(\cdot)$,   $\text{dim}(\cdot)$, $\text{Span}(\cdot)$, $| \cdot|$, and $\text{Rank}(\cdot)$ to represent the null space, dimension, span, cardinality, and rank of a matrix. $\| \cdot \|$ is the Euclidean norm of a vector or the spectral norm of a matrix. $\text{det}(\cdot)$ is the determinant of a square matrix. Denote $\otimes$ as the Kronecker product.
${I}_d$ stands for the identity matrix of dimension $d \times d$, and $\mathbf{1}_d$ and $\mathbf{0}_d$ the d-dimensional column vectors with all entries equal to $1$ and $0$, respectively.
An undirected graph is denoted by $\mathcal{G}=\{ \mathcal{V},\mathcal{E}\}$ consisting of a non-empty  node set $\mathcal{V}=\{1,  \cdots, n \}$ and an undirected edge set $\mathcal{E}  \subseteq \mathcal{V} \times \mathcal{V}$. $(i,j) \in \mathcal{E}$ is an undirected edge. $\kappa$ is a complete graph with $n$ nodes.

\subsection{Sensor Measurements}

{The nodes in a network are divided into two categories: 1) anchor nodes whose positions are known; 2) free nodes whose positions need to be determined. Let $\Sigma_g$ be a common \textit{global coordinate frame} in $\mathbb{R}^3$. Each node $i$ (including both anchor node and free node)} holds an unknown fixed \textit{local coordinate frame} $\Sigma_i$. Define $Q \in SO(3)$ as the 3-dimensional rotation matrix, where $SO(3)= \{ Q \in \mathbb{R}^{3 \times 3}: Q^TQ = QQ^T=I_3, \text{det}(Q)=1\}$. Let $Q_i$ be the unknown rotation matrix from $\Sigma_i$ to $\Sigma_g$. Denote $p_i, p_j \in \mathbb{R}^3$ as the positions of node $i$ and node $j$ in $\Sigma_g$.
Define
\begin{equation}\label{local}
e_{ij} = p_j -p_i = Q_i e_{ij}^{i}, \ \ g_{ij}= \frac{p_j -p_i}{d_{ij}}= Q_ig_{ij}^i, \ \ e_{ij}^{i} =  d_{ij}g_{ij}^i,
\end{equation}
where $e_{ij} \in \mathbb{R}^3$ is the relative position in $\Sigma_g$. $e_{ij}^{i}$ is the local relative position in $\Sigma_i$.
$g_{ij} \in \mathbb{R}^3$ is the relative bearing of $p_j$ with respect to $p_i$ in $\Sigma_g$. $g_{ij}^{i}$ is the local version of relative bearing in $\Sigma_i$.
$ d_{ij}=\|p_i-p_j \| \in \mathbb{R}$ is the distance between node $i$ and node $j$. $g_{ij}^Tg_{ik}={g^{i}_{ij}}^Tg^{i}_{ik}$ is the angle between edges $e_{ij}$ and $e_{ik}$. Next, we will introduce two kinds of constraints in a network: angle constraint and displacement constraint. 

\begin{figure}[t]
\centering
\includegraphics[width=1\linewidth]{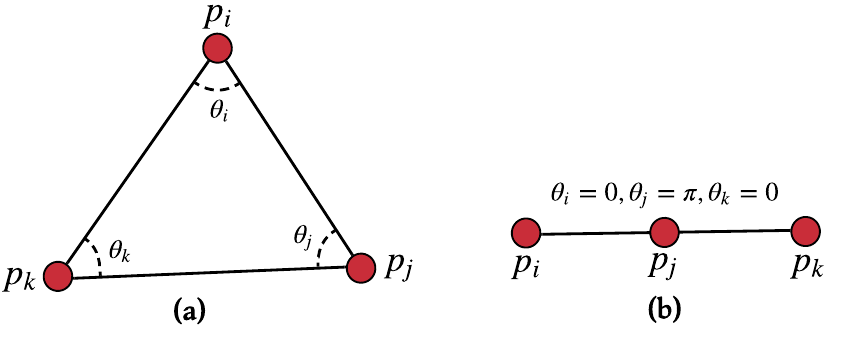}
\caption{Illustration of the angle constraints in Lemma \ref{motia} and Theorem \ref{angle12}. }
\label{motia}
\end{figure}

\subsection{Angle Constraint}\label{anc}

In this subsection, we will show that the angle constraints can be described by a group of parameters.

\begin{lemma}\label{angle}
For any three different nodes $p_i,p_j,p_k$ in $\mathbb{R}^3$ shown in Fig. \ref{motia},
there exist parameters $w_{ik}, w_{ki}, w_{ij}, w_{ji}, w_{jk}, w_{kj} \in \mathbb{R}$ such that
\begin{align}\label{ang1}
 &w_{ik}e_{ik}^Te_{ij}+w_{ki}e_{ki}^Te_{kj}=0, \\
  \label{aa1}
 &w_{ij}e_{ij}^Te_{ik}+w_{ji}e_{ji}^Te_{jk}=0, \\
 \label{aa2}
 &w_{jk}e_{jk}^Te_{ji}+w_{kj}e_{kj}^Te_{ki}=0,
\end{align}
with $w_{ik}^2+w_{ki}^2 \neq 0$, $w_{ij}^2+w_{ji}^2 \neq 0$, and $w_{jk}^2+w_{kj}^2 \neq 0$. 
\end{lemma}
\begin{proof}
Denote the angles
between $e_{ij}$ and $e_{ik}$, $e_{ji}$ and $e_{jk}$,  $e_{ki}$ and $e_{kj}$ by $\theta_i, \theta_j, \theta_k \in [0, \pi]$, respectively.
Note that $e_{ik}^Te_{ij}=d_{ik}d_{ij}\cos \theta_i$ and $e_{ki}^Te_{kj}=d_{ik}d_{jk}\cos \theta_k$.
To ensure  $w_{ik}e_{ik}^Te_{ij}+w_{ki}e_{ki}^Te_{kj}=0$ and $w_{ik}^2+w_{ki}^2 \neq 0$, the parameters $w_{ik}$ and $w_{ki}$ can be designed as
\begin{equation}\label{para}
\begin{array}{ll}
&  \! \begin{array}{lll} w_{ik} = \frac{1}{e_{ik}^Te_{ij}},  w_{ki} = \frac{-1}{e_{ki}^Te_{kj}}, 
& e_{ik}^Te_{ij}, e_{ki}^Te_{kj} \neq 0, \\ 
w_{ik} = 1,  w_{ki} = 0, & 
e_{ik}^Te_{ij}=0, e_{ki}^Te_{kj} \neq 0, \\ 
w_{ik} = 0,  w_{ki} = 1, & 
e_{ik}^Te_{ij} \neq 0,   e_{ki}^Te_{kj} = 0,\\
w_{ik} = 1,  w_{ki} = 1, & 
e_{ik}^Te_{ij} = 0, e_{ki}^Te_{kj} = 0.
\end{array}
\end{array} \\
\end{equation}

Similarly, the parameters $w_{ij}$, $w_{ji}$, $w_{jk}$, $w_{kj}$ in \eqref{aa1} and \eqref{aa2} can also be designed. 
\end{proof}

\begin{defn} \textbf{(Angle Constraint)}
Equations \eqref{ang1},  \eqref{aa1}, and \eqref{aa2} are formally defined as angle constraints for the nodes $i, j, k$  w.r.t. edges $(i,j), (i,k), (j,k)$. 
\end{defn}

\begin{remark}
In network localization, the angle constraints are only constructed among the anchor nodes by the known anchor positions. If there is an undirected edge between anchor node $i$ and anchor node $j$,
their known positions are shared by communication, i.e., $e_{ij}=p_i\!-\!p_j$
is available for the anchor nodes $i,j$. Hence,
for the known anchor nodes $i, j, k$ w.r.t. undirected edges $(i,j), (i,k), (j,k)$, $e_{ij}, e_{ik}, e_{jk}$ are available for the anchor nodes $i,j,k$ and their corresponding angle constraints \eqref{ang1}-\eqref{aa2} can be constructed by \eqref{para}. 
\end{remark}

\begin{theorem}\textbf{(Property of Angle Constraint)}\label{angle12}
For any three different nodes $p_i,p_j,p_k$
in $\mathbb{R}^3$ shown in Fig. \ref{motia}, denote the angles 
between $e_{ij}$ and $e_{ik}$, $e_{ji}$ and $e_{jk}$,  $e_{ki}$ and $e_{kj}$ by $\theta_i, \theta_j, \theta_k \in [0, \pi]$, respectively.  The angles $\theta_i, \theta_j, \theta_k \in [0, \pi]$ are determined uniquely by the parameters $w_{ik}, w_{ki}, w_{ij}, w_{ji}, w_{jk}, w_{kj}$ in the angle constraints \eqref{ang1}-\eqref{aa2}.
\end{theorem}

The proof of Theorem \ref{angle12} is given in \cite{fangtechnique}.
Note that Lemma \ref{angle} and \text{Theorem \ref{angle12}} are also applicable to the case that $p_i, p_j, p_k$ are on a line. Next, we will discuss the construction of displacement constraint based on different local relative measurements (local relative position, distance, local relative bearing, angle, or ratio-of-distance measurements).

\subsection{Displacement Constraint}\label{dissec}
Before exploring the displacement constraint, the following assumption is given.
\begin{assumption}\label{as3}
No two nodes are collocated in $\mathbb{R}^3$. For the five nodes $i,j,k,h,l$, node $i$
is the common neighboring node of nodes $j,k,h,l$.
\end{assumption}

For the node $i$ and its neighbors $j, k, h, l$ in $\mathbb{R}^3$,
the matrix $E_i=(e_{ij}, e_{ik}, e_{ih}, e_{il})$ $ \in \mathbb{R}^{3 \times 4}$ is a wide matrix. From the matrix theory,
there exists a non-zero vector $ \mu_i=(\mu_{ij}, \mu_{ik}, \mu_{ih}, \mu_{il})^T \in \mathbb{R}^4$ such that $E_i\mu_i= \mathbf{0}$, i.e.,

\begin{equation}\label{root}
     \mu_{ij}e_{ij}+\mu_{ik}e_{ik}+\mu_{ih}e_{ih} + \mu_{il}e_{il}= \mathbf{0},
\end{equation}
where $\mu_{ij}^2+\mu_{ik}^2+\mu_{ih}^2+\mu_{il}^2 \neq 0$.

\begin{defn}\textbf{(Displacement Constraint)}
Equation \eqref{root} is formally defined as a displacement constraint for node $i$ w.r.t. edges $(i,j), (i,k), (i,h), (i,l)$.
\end{defn}

Next, we will obtain the
parameters $\mu_{ij}, \mu_{ik}, \mu_{ih}, \mu_{il}$ in \eqref{root}
by using local relative position measurements.
Denote 
$e_{ij}^i, e_{ik}^i, e_{ih}^i, e_{il}^i$ as
the local relative positions measured by node $i$ in $\Sigma_i$.
Since $e_{ij} = Q_ie_{ij}^{i}, e_{ik} = Q_ie_{ik}^{i},e_{ih} = Q_ie_{ih}^{i}, e_{il} = Q_ie_{il}^{i}$,  \eqref{root} becomes
\begin{equation}\label{element}
Q_i  \left[ \!
\begin{array}{c c c c}
e_{ij}^{i} & e_{ik}^{i}  & e_{ih}^{i} &  e_{il}^{i} \\
\end{array}
\right]  \left[ \!
	\begin{array}{c}
	\mu_{ij} \\
	\mu_{ik} \\
	\mu_{ih} \\
	\mu_{il}
	\end{array}
	\right] = \mathbf{0}.
\end{equation}

Note that $Q_i^TQ_i=I_3$. Although the rotation matrix $Q_i$ is unknown, the non-zero vector $(\mu_{ij}, \mu_{ik}, \mu_{ih}, \mu_{il})^T$ in \eqref{element} can be obtained by solving an equivalent equation of \eqref{element} shown as
\begin{equation}\label{wmi}
\left[ \!
\begin{array}{c c c c}
e_{ij}^{i} & e_{ik}^{i}  & e_{ih}^{i} &  e_{il}^{i} \\
\end{array}
\right]  \left[ \!
	\begin{array}{c}
	\mu_{ij} \\
	\mu_{ik} \\
	\mu_{ih} \\
	\mu_{il}
	\end{array}
	\right] = \mathbf{0}.
\end{equation}

In a local-relative-position-based network in $\mathbb{R}^3$ under \text{Assumption \ref{as3}}, 
let 
$\mathcal{X}_{\mathcal{G}}=\{ ( i, j, k, h, l) \in \mathcal{V}^{5} : (i,j), (i,k), $ $ (i,h),  (i,l)  \in \mathcal{E},   j \!<\! k \!<\! h \!<\! l\}$. Each element of $\mathcal{X}_{\mathcal{G}}$ can be used to construct a local-relative-position-based displacement constraint \eqref{root}.

In addition, the displacement constraint \eqref{root} can also be obtained by distance, local relative bearing, angle, or ratio-of-distance measurements. The details are given in \cite{fangtechnique}. After introducing the angle constraint and displacement constraint, we are ready to establish the direct similarity of two
networks by 
using angle and displacement constraints.

\begin{figure}[t]
\centering
\includegraphics[width=1\linewidth]{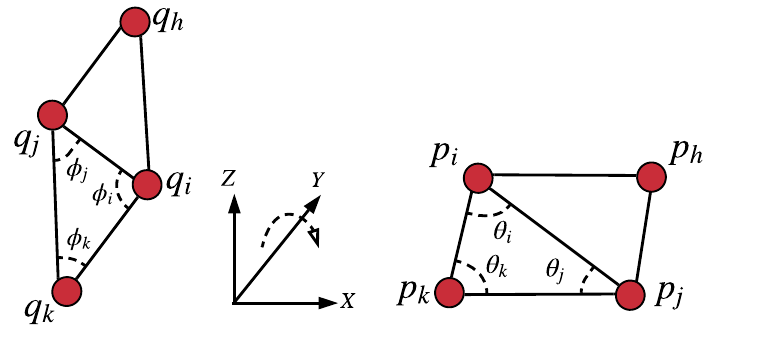}
\caption{ Two similar quadrilaterals. The network $\diamondsuit_{ijkh}(p)$ can be obtained by rotating $\diamondsuit_{ijkh}(q)$ around Y-axis by an angle $\frac{\pi}{2}$.}
\label{moti}
\end{figure}


\subsection{{Angle and Displacement Constraints based Similar Networks} }

Suppose that we aim to prove the direct similarity of two quadrilaterals $\diamondsuit_{ijkh}(p)$ and $\diamondsuit_{ijkh}(q)$ in $\mathbb{R}^3$ shown in Fig. \ref{moti}, where $p_i, p_j, p_k, p_h$ and $q_i, q_j, q_k, q_h$ are on two different planes, respectively.
Denote by $e'_{hi}\!=\!q_i\!-\!q_h, e'_{hk}\!=\!q_k\!-\!q_h, e'_{hj}\!=\!q_j\!-\!q_h$.
It is easy to prove that they are similar if the following conditions $\text{(a)-(b)}$ hold
\begin{enumerate}[(a)]
\item $ \mu_{hi}\!+\!\mu_{hj}\!+\!\mu_{hk} \neq 0$, $\mu_{hi}e_{hi}\!+\!\mu_{hj}e_{hj}\!+\!\mu_{hk}e_{hk}=\mathbf{0}$, 
$ \mu_{hi}e'_{hi}\!+\!\mu_{hj}e'_{hj}\!+\!\mu_{hk}e'_{hk}\!=\!\mathbf{0}$;
\item $\theta_{i}=\phi_{i}$, $\theta_{j}=\phi_{j}$, $\theta_{k}=\phi_{k}$.
\end{enumerate}

From \text{Theorem \ref{angle12}}, we can know that the angles $\theta_i, \theta_j, \theta_k$ are determined uniquely by the angle constraints \eqref{ang1}-\eqref{aa2}. Hence,
two quadrilaterals $\diamondsuit_{ijkh}(p)$ and $\diamondsuit_{ijkh}(q)$ are similar if the following conditions $\text{(a)-(e)}$ hold
\begin{enumerate}[(a)]
\item $ \mu_{hi}+\mu_{hj}+\mu_{hk} \neq 0$, $\mu_{hi}e_{hi}+\mu_{hj}e_{hj}+\mu_{hk}e_{hk}=\mathbf{0}$, $ \mu_{hi}e'_{hi}+\mu_{hj}e'_{hj}+\mu_{hk}e'_{hk}=\mathbf{0}$;
\item $w_{ik}e_{ik}^Te_{ij}\!+\!w_{ki}e_{ki}^Te_{kj}\!=\!0$,  $w_{ik}{e'_{ik}}^Te'_{ij}\!+\!w_{ki}{e'_{ki}}^Te'_{kj}\!=\!0$;
\item $w_{ij}e_{ij}^Te_{ik}\!+\!w_{ji}e_{ji}^Te_{jk}\!=\!0$,  $w_{ij}{e'_{ij}}^Te'_{ik}\!+\!w_{ji}{e'_{ji}}^Te'_{jk}\!=\!0$;
\item $w_{jk}e_{jk}^Te_{ji}\!+\!w_{kj}e_{kj}^Te_{ki}\!=\!0$,  $w_{jk}{e'_{jk}}^Te'_{ji}\!+\!w_{kj}{e'_{kj}}^Te'_{ki}\!=\!0$;
\item $w_{ik}^2+w_{ki}^2 \neq 0$, $w_{ij}^2+w_{ji}^2 \neq 0$, $w_{jk}^2+w_{kj}^2 \neq 0$.
\end{enumerate}

\begin{remark}
From \text{Theorem \ref{angle12}}, the conditions $\text{(b)-(e)}$ are equivalent to the conditions $\theta_{i}=\phi_{i}$, $\theta_{j}=\phi_{j}$, $\theta_{k}=\phi_{k}$.
\end{remark}

Note that $\mu_{hi}e_{hi}+\mu_{hj}e_{hj}+\mu_{hk}e_{hk}=\mathbf{0}$ is a displacement constraint. $w_{ik}e_{ik}^Te_{ij}\!+\!w_{ki}e_{ki}^Te_{kj}\!=\!0$, $w_{ij}e_{ij}^Te_{ik}\!+\!w_{ji}e_{ji}^Te_{jk}\!=\!0$, and $w_{jk}e_{jk}^Te_{ji}\!+\!w_{kj}e_{kj}^Te_{ki}\!=\!0$ are three angle constraints. Therefore,
two quadrilaterals $\diamondsuit_{ijkh}(p)$ and $\diamondsuit_{ijkh}(q)$ are similar if a displacement constraint $\mu_{hi}e_{hi}+\mu_{hj}e_{hj}+\mu_{hk}e_{hk}=\mathbf{0}$ with $\mu_{hi}+\mu_{hj}+\mu_{hk} \neq 0$ and three angle constraints \eqref{ang1}-\eqref{aa2}  for the nodes $i,j,k,h$ w.r.t. $p$ are equal to those w.r.t. $q$.

In this paper, we are interested in generalizing the above mentioned angle and displacement constraints in quadrilaterals to a set of angle and displacement constraints in more general networks in $\mathbb{R}^3$. The immediate question is whether a set of angle and displacement constraints can uniquely characterize a network up to directly similar transformations, i.e., translation, rotation, and scaling. The following sections provide some insights into this question.

\section{Angle-Displacement Rigidity Theory}\label{hyper}

A network in $\mathbb{R}^3$ is denoted by $(\mathcal{G}$, $p)$, where  $\mathcal{G}=\{ \mathcal{V},\mathcal{E}\}$ is an undirected graph and $p=(p_1^T, \cdots, p_{n}^T)^T \in \mathbb{R}^{3n}$ is a configuration in $\mathcal{G}$.
Denote by $|\mathcal{X}_{\mathcal{G}}|=m_d$ ($\mathcal{X}_{\mathcal{G}}$ is defined in Section \ref{dissec} and \cite{fangtechnique}).
Let $\Upsilon_{\mathcal{G}}=\{ ( i, j, k) \in \mathcal{V}^{3} : (i,j), (i,k), (j,k) \in \mathcal{E}, i< \! j \!<\! k \}$ with $|\Upsilon_{\mathcal{G}}|=m_r$. In this section, we focus on when the angle and displacement constraints can uniquely characterize a network in $\mathbb{R}^3$. The following function is first introduced.

For the angle constraints, the \textit{angle function} $B_{\Upsilon_{\mathcal{G}}}(p) : \mathbb{R}^{3n} \rightarrow \mathbb{R}^{m_r}$ is defined as
\begin{equation}\label{funa}
B_{\Upsilon_{\mathcal{G}}}(p)
\!=\!  (\cdots, w_{ik}e_{ik}^Te_{ij}+w_{ki}e_{ki}^Te_{kj}, \cdots)^T, 
\end{equation}
where $(i,j,k) \! \in \! \Upsilon_{\mathcal{G}}$, and $w_{ik} , w_{ki}$ are obtained by \eqref{para} which satisfy the condition
$w_{ik}e_{ik}^Te_{ij}\!+\!w_{ki}e_{ki}^Te_{kj}\!=\!0$ with $w_{ik}^2\!+\!w_{ki}^2 \! \neq \! 0$.

For the displacement constraints, the \textit{displacement function} $L_{\mathcal{X}_{\mathcal{G}}}(p) : \mathbb{R}^{3n} \rightarrow \mathbb{R}^{3m_d}$
is defined as 
\begin{equation}\label{func}
L_{\mathcal{X}_{\mathcal{G}}}(p)
\!=\! (\cdots, \mu_{ij}e_{ij}^T\!+\! \mu_{ik}e_{ik}^T\!+\! \mu_{ih}e_{ih}^T\!+\! \mu_{il}e_{il}^T, \cdots)^T,
\end{equation}
where $( i, j, k, h, l) \in \mathcal{X}_{\mathcal{G}}$ and $\mu_{ij}e_{ij}+\mu_{ik}e_{ik}+\mu_{ih}e_{ih} + \mu_{il}e_{il}= \mathbf{0}$ with $\mu_{ij}^2+\mu_{ik}^2+\mu_{ih}^2+\mu_{il}^2 \neq 0$.

Let $\mathcal{T}_{\mathcal{G}} = \Upsilon_{\mathcal{G}} \cup \mathcal{X}_{\mathcal{G}} $ and $m=3m_d+m_r$. Combining angle constraints and displacement constraints,
the \textit{angle-displacement function} $f_{\mathcal{T}_{\mathcal{G}}}(p) : \mathbb{R}^{3n} \rightarrow \mathbb{R}^{m}$ is defined as
\begin{equation}\label{cust}
    f_{\mathcal{T}_{\mathcal{G}}}(p) = (B_{\Upsilon_{\mathcal{G}}}^T(p), L^T_{\mathcal{X}_{\mathcal{G}}}(p)  )^T.
\end{equation}

We are now ready to define the fundamental concepts in angle-displacement rigidity theory. These concepts are defined analogously to those in the distance rigidity theory \cite{aspnes2006theory}, bearing rigidity theory \cite{zhao2015bearing}, and angle rigidity theory \cite{jing2019angle, jing2019multi}.
{Our proposed angle constraint $w_{ik}e_{ik}^Te_{ij}+w_{ki}e_{ki}^Te_{kj}=0$ with $w_{ik}^2+w_{ki}^2 \neq 0$ contains the relationship between $\theta_i$ and $\theta_k$, which
is different from the angle constraint $g_{ij}^Tg_{ik}=\cos \theta_i$ in \cite{jing2019angle} that only contains angle information $\theta_i$.
An advantage of the proposed angle constraint is that 
it can be used for analyzing the network localizability in 3-D space shown in Section \ref{fomu}}. 

\begin{remark}
 For two different configurations  $p$ and $p'$, their corresponding parameters $w_{ik}, w_{ki}$ in \eqref{funa} and $\mu_{ij}, \mu_{uk}, \mu_{ih}, \mu_{il}$ in \eqref{func} may be different. If the function $f_{\mathcal{T}_{\mathcal{G}}}(\cdot)$ in \eqref{cust} is customized
for the configuration $p$,
we have $f_{\mathcal{T}_{\mathcal{G}}}(p)= \mathbf{0}$, but $f_{\mathcal{T}_{\mathcal{G}}}(p')$ may not be a zero vector. It is easy to prove that if the configuration $p'$ satisfies $ p'=s(I_n \otimes  {Q})p +   \mathbf{1}_n\otimes \mathbf{t}, s \in \mathbb{R} \backslash \{0\}, {Q} \in SO(3),  \mathbf{t} \in \mathbb{R}^3$, we have $f_{\mathcal{T}_{\mathcal{G}}}(p')=\mathbf{0}$.
In addition, the reflection ambiguity \cite{chen2020} can be avoided by combining the angle constraint and displacement constraint.
\end{remark}

The \textit{angle-displacement rigidity matrix} is defined as
\begin{equation}\label{rigiditym}
R(p) = \frac{\partial f_{\mathcal{T}_{\mathcal{G}}}(p)}{\partial p} \in \mathbb{R}^{m \times 3n}. 
\end{equation}

An angle-displacement infinitesimal motion is a motion preserving the invariance of $f_{\mathcal{T}_{\mathcal{G}}}(p)$. In other words, an angle-displacement infinitesimal motion $\delta p$ satisfies $R(p)\delta p= 0$. 
Since  $ \mu_{ij}e_{ij}\!+\! \mu_{ik}e_{ik}\!+\! \mu_{ih}e_{ih}\!+\! \mu_{il}e_{il}=\mathbf{0}$ and $w_{ik}e_{ik}^Te_{ij}+w_{ki}e_{ki}^Te_{kj}=0$, the angle-displacement infinitesimal motions preserving the invariance of $f_{\mathcal{T}_{\mathcal{G}}}(p)$ include translations, rotations and scalings. 

\begin{defn}
An angle-displacement infinitesimal motion is called trivial if it corresponds to a translation, a rotation and a scaling of the entire network.
\end{defn}

\begin{defn}\label{cong}
(\textbf{Angle-displacement Equivalency and Congruency}) Let ($\kappa$, $p$) be the complete network, where 
$\kappa$ is the complete graph and $R^{\kappa}(p)$ is the corresponding angle-displacement rigidity matrix.  Two networks ($\mathcal{G}$, $p$) and ($\mathcal{G}$, $p'$) are angle-displacement equivalent if $R(p)p'=0$, and angle-displacement congruent if $R^{\kappa}(p)p'=0$.
\end{defn}

\begin{defn}
(\textbf{Angle-displacement Rigidity}) A network ($\mathcal{G}$, $p$) is angle-displacement rigid if
there exists a constant $\epsilon>0$ such that any network ($\mathcal{G}$, $p'$) that is angle-displacement equivalent to it and satisfies $\|p-p' \| < \epsilon$ is also angle-displacement congruent to it.
\end{defn}

\begin{defn}\label{grid}
(\textbf{Globally Angle-displacement Rigidity}) A network ($\mathcal{G}$, $p$) is globally angle-displacement rigid if it is angle-displacement congruent to any of its angle-displacement equivalent networks.
\end{defn}

\subsection{Properties of Angle-displacement Rigidity}

We next explore some important properties of  angle-displacement rigidity.

\begin{defn}
A network ($\mathcal{G}$, $p$) is {infinitesimally angle-displacement rigid} in $\mathbb{R}^3$ if all the angle-displacement infinitesimal motions satisfying $\frac{ \partial f_{\mathcal{T}_{\mathcal{G}}}(p)}{\partial p} \delta p = 0 $ are {trivial}, i.e., $\text{dim}(\text{Null}(R(p)))=7$. 
\end{defn}

\begin{lemma}\label{inva}
The trivial motion space for angle-displacement rigidity in $\mathbb{R}^3$ is $S=S_t \cup S_r \cup S_s$, where $S_t=\{ \mathbf{1}_n\otimes {I}_3 \}$ is the space including all infinitesimal motions that correspond to 
translational motions, $S_r= \{(I_n \otimes A)p, A+A^T =\mathbf{0}, A \in \mathbb{R}^{3 \times 3} \}$ is the space including all infinitesimal motions that correspond to rotational motions, and $S_s= \text{Span}(p)$ is the space including all infinitesimal motions that correspond to scaling motions.
\end{lemma}

\begin{proof}
Let 
\begin{equation}
    \begin{array}{ll}
         &  \eta_d^T = \frac{\partial (\mu_{ij}e_{ij}\!+\! \mu_{ik}e_{ik}\!+\! \mu_{ih}e_{ih}\!+\! \mu_{il}e_{il}) }{\partial p} ,\\
         &  \eta_r^T = \frac{\partial (w_{ik}e_{ik}^Te_{ij}+w_{ki}e_{ki}^Te_{kj}) }{\partial p},
    \end{array}
\end{equation}
where $\eta_d^T$ and $\eta_r^T$ are the rows of $R(p)$. Then, we have
\begin{equation}\label{infid}
 \eta_d\!=\![\mathbf{0}, 	\!-\!\mu_{ij}\!-\!\mu_{ik}\!-\!\mu_{ih}\!-\!\mu_{il}, \mathbf{0}, \mu_{ij}, \mathbf{0}, \mu_{ik}, \mathbf{0}, \mu_{ih}, \mathbf{0}, \mu_{il},  \mathbf{0}]^T \otimes I_3.   
\end{equation}

\begin{equation}\label{infi}
	\eta_r \!=\! \left[ \!
	\begin{array}{c}
	\mathbf{0} \\
	2w_{ik}p_i+ (w_{ki}-w_{ik})p_j-(w_{ik}+w_{ki})p_k\\
	\mathbf{0} \\
	(w_{ki}-w_{ik})p_i+(w_{ik}-w_{ki})p_k \\
	\mathbf{0} \\
	-(w_{ik}+w_{ki})p_i+(w_{ik}-w_{ki})p_j+2w_{ki}p_k\\
	\mathbf{0} 
	\end{array}
	\right].
\end{equation}

Since $ \mu_{ij}e_{ij}\!+\! \mu_{ik}e_{ik}\!+\! \mu_{ih}e_{ih}\!+\! \mu_{il}e_{il}=\mathbf{0}$ and $w_{ik}e_{ik}^Te_{ij}+w_{ki}e_{ki}^Te_{kj}=0$,
for the scaling space $S_s$,  it is straightforward that $\eta_d^Tp =\mathbf{0}$ and $\eta_r^Tp =0$. For the translation space $S_t$, we have $\eta_d^T( \mathbf{1}_n\otimes {I}_3)=\mathbf{0}$ and $\eta_r^T( \mathbf{1}_n\otimes {I}_3)=\mathbf{0}$.  For the rotation space $S_r= \{(I_n \otimes A)p, A+A^T =\mathbf{0}, A \in \mathbb{R}^{3 \times 3} \}$, it follows that $\eta_d^T (I_n \otimes A) p = A (\mu_{ij}e_{ij}\!+\! \mu_{ik}e_{ik}\!+\! \mu_{ih}e_{ih}\!+\! \mu_{il}e_{il})=\mathbf{0}$ and 
\begin{equation}
\begin{array}{ll}
\eta_r^T (I_n \otimes A) p\\
= w_{ik}p_i^T(A+A^T)p_i \!+\! (w_{ki}\!-\!w_{ik})p_j^T(A\!+\!A^T)p_i 
\\
-(w_{ik}\!+\!w_{ki})p_k^T(A\!+\!A^T)p_i + (w_{ik}\!-\!w_{ki})p_k^T(A\!+\!A^T)p_j \\
+w_{ki}p_k^T(A\!+\!A^T)p_k =0.
\end{array}
\end{equation}

Then, the conclusion follows.
\end{proof}

\begin{theorem}
(\textbf{Property of Infinitesimally Angle-displacement Rigid})
Let $\kappa$ be a complete graph with $n$ nodes.
If a network ($\kappa$, $p$) is infinitesimally angle-displacement rigid in $\mathbb{R}^3$, then $f_{\kappa}^{-1}(f_{\kappa}(p)) = \{ s(I_n \otimes  {Q})p +   \mathbf{1}_n\otimes \mathbf{t}, s \in \mathbb{R} \backslash \{0\}, {Q} \in SO(3),  \mathbf{t} \in \mathbb{R}^3\}$, which is a $7$-dimensional manifold. 
\end{theorem}
\begin{proof}
	For any $q \in f_{\kappa}^{-1}(f_{\kappa}(p)), q=s(I_n \otimes  {Q})p +   \mathbf{1}_n\otimes \mathbf{t}, s \in \mathbb{R} \backslash \{0\}, {Q} \in SO(3),  \mathbf{t} \in \mathbb{R}^3$. From the chain rule, it yields
	\begin{equation}
	\begin{array}{ll}
\frac{\partial f_{\kappa}(q)}{\partial q} = 
\frac{\partial f_{\kappa}(p)}{\partial sp}(I_n \otimes Q^T).
	\end{array}
	\end{equation}
	
Then, we have $\text{Rank}(\frac{\partial f_{\kappa}(q)}{\partial q})=\text{Rank}(\frac{\partial f_{\kappa}(p)}{\partial sp}(I_n \otimes {Q}^T))=3n-7$. Since $ f_{\kappa} : \mathbb{R}^{3n} \rightarrow \mathbb{R}^{|\mathcal{T}_{\mathcal{\kappa}}|}$ is a smooth mapping, $f_{\kappa}^{-1}(f_{\kappa}(p))$ is a properly embedded manifold of dimension $3n - (3n\!-7)=7$.
\end{proof}

\begin{lemma}\label{grid1}
A  network ($\mathcal{G}$, $p$) is globally angle-displacement rigid if $\text{Null}(R^{\kappa}(p)) = \text{Null}(R(p))$. 
\end{lemma}
\begin{proof}
From Definition \ref{cong},
any network ($\mathcal{G}$, $p'$) that is equivalent to ($\mathcal{G}$, $p$) satisfies $R(p)p'=0$. It then follows from $\text{Null}(R^{\kappa}(p)) = \text{Null}(R(p))$ that $R^{\kappa}(p)p'=0$, which means that ($\mathcal{G}$, $p'$) is also congruent to ($\mathcal{G}$, $p$). From 
Definition \ref{grid}, we can know that ($\mathcal{G}$, $p$) is globally angle-displacement rigid.
\end{proof}

Next, a relationship between infinitesimally angle-displacement rigid and globally angle-displacement rigid is given.

\begin{theorem} If a network ($\mathcal{G}$, $p$) is infinitesimally angle-displacement rigid in $\mathbb{R}^3$, it is globally angle-displacement rigid.
\end{theorem}
\begin{proof}

Similar to \textit{Lemma \ref{inva}}, it can be proved that $\text{Span}\{ \mathbf{1}_n\otimes I_3, p, (I_n \otimes A)p, A+A^T =\mathbf{0},  A \subseteq \mathbb{R}^{3 \times 3} \} \subseteq \text{Rank}(R^{\kappa}(p))$.
For any $\delta p \in \text{Null}(R^{\kappa}(p))$, we have $R^{\kappa}(p)(p+\delta p)=0$.
Since graph $\mathcal{T}_{\mathcal{G}}$ is a subgraph of graph $\mathcal{T}_\kappa$, we have $R(p)(p+\delta p)=0$ and hence $R(p)\delta p=0$. Then, we have $\text{Null}(R^{\kappa}(p)) \subseteq \text{Null}(R(p))$. Hence, $\text{Span}\{ \mathbf{1}_n\otimes I_3, p, (I_n \otimes A)p, A+A^T =\mathbf{0},  A \subseteq \mathbb{R}^{3 \times 3} \}\subseteq \text{Null}(R^{\kappa}(p)) \subseteq \text{Null}(R(p))$. If a network ($\mathcal{G}$, $p$) is infinitesimally angle-displacement rigid, we have $\text{Null}(R(p))= \text{Span}\{ \mathbf{1}_n\otimes I_3, p, (I_n \otimes A)p, A+A^T =\mathbf{0},  A \subseteq \mathbb{R}^{3 \times 3} \}$, which means that $\text{Null}(R^{\kappa}(p)) = \text{Null}(R(p))$. From \textit{Lemma \ref{grid1}}, we can know that ($\mathcal{G}$, $p$) is globally angle-displacement rigid.
\end{proof}

\subsection{Analysis of Angle and Displacement Constraints from the Angle-displacement Rigidity Theory}\label{reason}

The angle-displacement rigidity matrix $R(p)$ in \eqref{rigiditym} will be used for analyzing the network localizability in Section \ref{fomu} and designing distributed network localization protocols in Section \ref{distri}. Hence, we need to construct an available angle-displacement rigidity matrix $R(p)$.  For an
angle constraint $w_{ik}e_{ik}^Te_{ij}\!+\!w_{ki}e_{ki}^Te_{kj}\!=\!0$, its corresponding rows 
$\eta_r^T$ of $R(p)$ in \eqref{infi} contain the position information $p_i, p_j, p_k$. 
Since only the positions of the anchor nodes are known, to obtain the angle-displacement rigidity matrix $R(p)$, the angle constraints can only be constructed among the anchor nodes. The method to construct angle constraints among the anchor nodes is shown in \eqref{para}. 
For displacement constraint $ \mu_{ij}e_{ij}\!+\! \mu_{ik}e_{ik}\!+\! \mu_{ih}e_{ih}\!+\! \mu_{il}e_{il}=\mathbf{0}$, its corresponding row $\eta_d^T$ of $R(p)$ in \eqref{infid} does not contain any position information. Hence, the displacement constraints can be constructed among all the nodes. As introduced in Section \ref{dissec} and \cite{fangtechnique}, the local relative measurements (local relative position, distance, local relative bearing, angle, or ratio-of-distance measurements) can be used for constructing displacement constraints. 

\begin{remark}
Note that there already exist some works on rigidity theory with mixed types of geometric constraints \cite{kwon2018generalized,kwon2019hybrid,stacey2017role,  bishop2015distributed}.  
The nonlinearity of the constraints in \cite{kwon2018generalized,kwon2019hybrid,  bishop2015distributed} leads to relative-position-based distributed algorithms, which are not applicable to bearing-only-based, angle-only-based,  distance-only-based, or ratio-of-distance-based distributed network localization where the relative position measurements are unavailable to each node.
The rigidity theory in \cite{stacey2017role} can be used for solving the relative-position-based distributed network localization problem, but it needs global coordinate frame. The local-relative-position-based distributed network localization problem is not solved in \cite{stacey2017role}.
\end{remark}

\section{Angle-displacement rigidity theory based network localizability
}\label{fomu}

Denote $\mathcal{V}_a=\{ 1,\cdots, n_a \}$ as the anchor nodes set with $|\mathcal{V}_a|=n_a$, whose positions, denoted by $p_a = (p_1^T, \cdots, p_{n_a}^T)^T$,
are known. Denote $\mathcal{V}_f\!=\!\{ n_a \!+\! 1,  \cdots, n \}$ as
the free nodes set with $|\mathcal{V}_f|=n_f$,  whose positions, denoted by $p_f = (p_{n_a\!+\!1}^T, \cdots, p_{n}^T)^T$, 
need to be determined. Before presenting the main results, the following assumption made which is necessary for network localizability analysis.

\begin{assumption}\label{ass3}
No two nodes are collocated in $\mathbb{R}^3$. Each anchor node has at least two neighboring anchor nodes, and each free node has at least four neighboring nodes in $\mathbb{R}^3$. For  local-relative-bearing-based and angle-based networks, each free node and its neighbors are non-colinear. 
\end{assumption}

\begin{remark}
Note that if local relative position measurements are used to localize the network in the proposed method, the condition in
Assumption \ref{ass3} can be relaxed as: for a 
free node $i$ with less than four neighboring nodes, 
there is at least one neighboring node $j \in \mathcal{N}_i $ satisfying $|\mathcal{N}_j|\ge 4$, i.e., the node $j$ has at least four neighboring nodes $i, k, h, l$ (refer to Section \ref{dissec} and Assumption \ref{as3}). 
The assumption that each free node has at least $4$ neighboring nodes in $\mathbb{R}^3$ can also be found in the literature on network localization and formation control \cite{han2017barycentric,han2017tc, lin2015necessary}.
In addition, compared with the works
\cite{han2017barycentric,han2017tc, lin2015necessary} that require an additional generic assumption, e.g., any three nodes are not on the same line and any four nodes are not on the same plane, Assumption \ref{ass3} is mild. 
\end{remark}

In this section, the network localization problem is formulated under Assumption \ref{ass3}, and necessary and sufficient conditions will be provided for network localizability based on the proposed angle-displacement rigidity theory.

\subsection{3-D Network Localization}


The angle constraints among the anchor nodes are only used for analyzing the network localizability, while    
the displacement constraints among all the nodes are not only used for analyzing the network localizability, but also used for estimating the locations of free nodes. The problem of network localization is formally stated below. 

\begin{problem}
{(\textbf{Network Localization Problem)} Consider a local-relative-measurement-based network ($\mathcal{G}$, $p$) in $\mathbb{R}^3$ under Assumption \ref{ass3}.} The network localization is to determine the positions of the free nodes $p_f$, given the positions of the anchor nodes $p_a$ and displacement constraints in the graph $\mathcal{X}_{\mathcal{G}}$
\begin{equation}\label{pro}
	\begin{array}{ll}
& \mu_{ij}\hat e_{ij}\!+\! \mu_{ik}\hat e_{ik}\!+\! \mu_{ih}\hat e_{ih}\!+\! \mu_{il}\hat e_{il}= \mathbf{0}, \ \  ( i, j, k, h, l) \in \mathcal{X}_{\mathcal{G}}, \\
& \hat{p}_i =p_i, \ \ i \in \mathcal{V}_a,
\end{array}
\end{equation}
where  $ \hat p_i,  \hat p_{j}, \hat p_k, \hat p_h, \hat p_l$ are the estimates of $p_i, p_j, p_k, p_h, p_l$, and $\hat e_{ij}= \hat p_{j} \!-\! \hat p_i, \hat e_{ik}= \hat p_{k} \!-\! \hat p_i, \hat e_{ih}= \hat p_{h} \!-\! \hat p_i, \hat e_{il}= \hat p_{l} \!-\! \hat p_i$.
\end{problem} 

\begin{remark}
In our proposed localization approach, we only consider the same type of local relative measurements for all nodes. 
This same type of local relative measurements can be angle, distance, local relative bearing, local relative position, or ratio-of-distance measurements. As introduced in Section \ref{dissec} and \cite{fangtechnique}, the displacement constraint can be obtained by one same type local relative measurements. Hence, the network localization is implemented based on the known anchor positions and one same type of local relative measurements. 
\end{remark}


\begin{defn}
A network ($\mathcal{G}$, $p$) is called localizable if the solution $p$ to \eqref{pro} is unique. 
\end{defn}

The cost function of network localization is designed as
\begin{equation}\label{cost}
\begin{array}{ll}
& \hspace*{-0.5cm}    \min\limits_{\hat p \in \mathbb{R}^{3n}} J(\hat p) \!=\! \sum\limits_{ ( i,  j, k, h, l)  \in  \mathcal{X}_{\mathcal{G}}} \|\mu_{ij}\hat e_{ij}\!+\! \mu_{ik}\hat e_{ik}\!+\! \mu_{ih}\hat e_{ih}\!+\! \mu_{il}\hat e_{il}\|^2, \\
& \hspace*{-0.5cm} \text{subject} \  \text{to} \ \ \hat{p}_i =p_i, \ \ i \in \mathcal{V}_a,
\end{array}
\end{equation}
where $\hat p = (\hat p_1^T, \cdots, \hat p_n^T)^T$ is the estimate of $p$.
The immediate question is when the true location $p$ is the unique global minimizer of \eqref{cost}.
Based on \eqref{infid} and \eqref{infi}, we have
\begin{equation}
    R(p)p = (2B_{\Upsilon_{\mathcal{G}}}^T(p), L^T_{\mathcal{X}_{\mathcal{G}}}(p)  )^T \in \mathbb{R}^m.
\end{equation}

Then, we get
\begin{equation}\label{anre}
p^TR(p)^TR(p)p = 4\| B_{\Upsilon_{\mathcal{G}}}(p) \|^2 + J(p), 
\end{equation}
where $J(p)=\|L_{\mathcal{X}_{\mathcal{G}}}(p)\|^2$ according to \eqref{func}. As analyzed in Section \ref{reason}, the angle-displacement rigidity matrix $R(p)$ can be obtained by the known anchor positions and local relative measurements (local relative position, distance, local relative bearing, angle, or ratio-of-distance measurements), and $R(p)$ does not contain position information of the free nodes, i.e., $R(p) = R([\begin{array}{c}
    p_{a}  \\
    \mathbf{0}  
    \end{array}]) = R(\hat p)$. 
From \eqref{anre}, we have
\begin{equation}\label{anre1}
 4\| B_{\Upsilon_{\mathcal{G}}}(\hat p) \|^2 \!+\! J(\hat p)= \hat p^TR(\hat p)^TR(\hat p)\hat p = \hat p^TR(p)^TR(p)\hat p.  
\end{equation}

Since the angle constraints are only constructed among the known anchor nodes, we have $B_{\Upsilon_{\mathcal{G}}}(p) = B_{\Upsilon_{\mathcal{G}}}([\begin{array}{c}
    p_{a}  \\
    \mathbf{0}  
    \end{array}])=B_{\Upsilon_{\mathcal{G}}}(\hat p)=0$ from \eqref{funa}. Then, based on \eqref{anre1}, \eqref{cost} 
can be rewritten as an angle and displacement constraints based cost function
\begin{equation}\label{cost1}
\begin{array}{ll}
&  \min\limits_{\hat p \in \mathbb{R}^{3n}} J(\hat p) = 4\| B_{\Upsilon_{\mathcal{G}}}(\hat p) \|^2 \!+\! J(\hat p)=  \hat{p}^T R(p)^T R(p)  \hat{p}, \\
& \text{subject} \  \text{to} \ \ \hat{p}_i =p_i, \ \ i \in \mathcal{V}_a.
\end{array}
\end{equation}

Let $D = R(p)^T R(p)$ be angle-displacement
information matrix. Since the nodes 
are divided into anchor nodes $p_a$ and free nodes $p_f$, the angle-displacement information matrix $D$ can be partitioned as
\begin{equation}
    D = \left[\begin{array}{cc}
    D_{aa} & D_{af} \\
    D_{fa}  &  D_{ff}  
    \end{array}\right],
\end{equation}
where $D_{aa} \in \mathbb{R}^{3{n_a} \times 3{n_a}}$, $D_{fa}^T = D_{af} \in \mathbb{R}^{3{n_a} \times 3{n_f}}$, and $D_{ff} \in \mathbb{R}^{3{n_f} \times 3{n_f}}$.

\begin{remark}\label{ra1}
The angle function  $B_{\Upsilon_{\mathcal{G}}}(p)$ in \eqref{funa} only contains the angle constraints among the anchor nodes, i.e., the angle function $B_{\Upsilon_{\mathcal{G}}}(p)$ only influences $D_{aa}$. 
\end{remark}

Inspired by the work \cite{zhao2016localizability}, the unique global minimizer of \eqref{cost1}, algebraic condition, rigidity condition, and topological condition for localizability are given below.

\begin{lemma}\label{ls4}
For the cost function \eqref{cost1}, any minimizer $\hat p_f^*$ is also a global minimizer and satisfies
\begin{equation}
    D_{ff}\hat p_f^* + D_{fa}p_a = \mathbf{0}.
\end{equation}
\end{lemma}
\begin{proof}
The cost function \eqref{cost1} can be rewritten as 
\begin{equation}
\min\limits_{\hat p_f \in \mathbb{R}^{3n_f}} \tilde J(\hat p_f) =  \hat{p}_f^T D_{ff} \hat{p}_f+  2{p}_a^T D_{af} \hat{p}_f+ {p}_a^T D_{aa} {p}_a. 
\end{equation}

Any minimizer $\hat p_f^*$ must satisfy 
$\bigtriangledown_{{\hat p_f}^*} \tilde J({\hat p_f}^*)=D_{ff} {{\hat p_f}^*} + D_{fa}p_a=0$. Since $p \in \text{Null}(R(p))$, we have $p \in \text{Null}(D)$, i.e., $D_{ff} p_f + D_{fa}p_a=0$. Denote $\hat p_f^*=p_f+x$ where $x \in \text{Null}(D_{ff})$. Let $\hat p^*=[p_a^T, (\hat p_f^*)^T]^T$. 
Then, we have $J(\hat p^*)= x^TD_{ff}x=0$. Hence, the cost function $J({\hat p^*})$ equals $0$.
\end{proof}

\begin{figure}[t]
\centering
\includegraphics[width=1\linewidth]{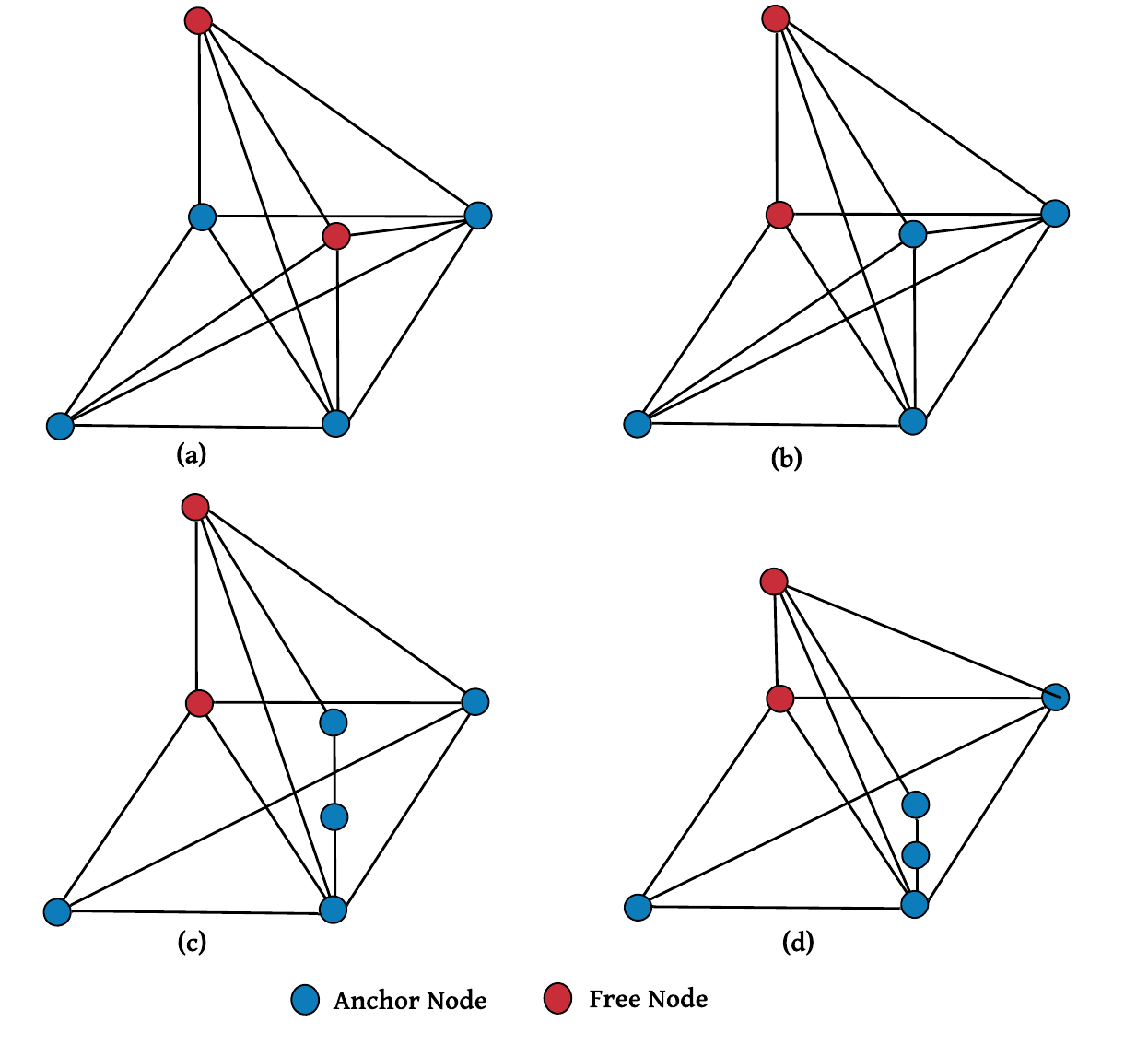}
\caption{Examples of 3-D networks. The blue dots are the anchor nodes and the red dots are the free nodes. The network in (a) is neither localizable nor infinitesimally angle-displacement rigid. The network in (b) is localizable and infinitesimally angle-displacement rigid.  The networks in (c) and (d) are localizable and have the same angle-displacement function, but they
are not infinitesimally angle-displacement rigid. }
\label{fig51}
\end{figure}

\begin{theorem}\label{the3}
(\textbf{Algebraic Condition for Localizability}) {Under Assumption \ref{ass3},} a local-relative-measurement-based network ($\mathcal{G}$, $p$) in $\mathbb{R}^3$ is localizable if and only if the matrix $D_{ff}$ is nonsingular. In this situation, the true positions of the free nodes can be obtained by ${p}_f = - D_{ff}^{-1}D_{fa}{p}_a$.
\end{theorem}
\begin{proof}
From \textit{Lemma \ref{ls4}}, any minimizer 
$\hat p_f^*$ must satisfy $D_{ff}\hat p_f^* + D_{fa}p_a=0$. Hence, it is clear that the minimizer $\hat p_f^*$ is unique if and only if $D_{ff}$ is nonsigular. When $D_{ff}$ is nonsigular, we have $\hat p_f^*=-D_{ff}^{-1}D_{fa}p_a$. Since $D_{ff}p_f + D_{fa}p_a=0$, $\hat p_f^*$ equals the true location $p_f$.
\end{proof}

\begin{remark}
The rank condition in Theorem \ref{the3} basically implies that a network is localizable if the positions of unknown free nodes can be obtained by the known anchor nodes. Note that a node in $\mathbb{R}^2$ or $\mathbb{R}^3$ can be represented by its three or four neighboring nodes, and the
displacement constraint can be used to describe the relationship among the node and its neighboring nodes. 
Hence, a node with a known displacement constraint involving its neighboring nodes and itself can determine its own position given the positions of its neighboring nodes. In other words, the free nodes may be determined
based on known anchor positions and displacement constraints.

Motivated by the above fact,
based on the displacement constraints and known anchor positions, we have $D_{ff}p_f + D_{fa}p_a = \mathbf{0}$. It is clear that $p_f$ is localizable if the matrix $D_{ff}$ is full rank, i.e., the positions of free nodes $p_f = -D_{ff}^{-1}D_{fa}{p}_a$ can be obtained by those of anchor nodes.
\end{remark}

\begin{remark}\label{rr2}
Since the angle-displacement
information matrix $D= R(p)^T R(p)$ is a symmetric positive semi-definite matrix, we can know that the matrix $D_{ff}$ is also a symmetric positive semi-definite matrix. If $D_{ff}$ is nonsigular, it must be a symmetric positive definite matrix.
\end{remark}

\begin{theorem}\label{the4}
(\textbf{Rigidity Condition for Localizability}) {Under Assumption \ref{ass3},} a local-relative-measurement-based network ($\mathcal{G}$, $p$) in $\mathbb{R}^3$ is localizable if and only if every angle-displacement infinitesimal  motion involves at least one anchor node; that is, for any nonzero angle-displacement infinitesimal motion $  \delta p = [\delta p_a^T,  \delta p_f^T ]^T \in \text{Null}(D) $,   $\delta p_a \neq \mathbf{0}$.
\end{theorem}
\begin{proof}
From \text{Theorem \ref{the3}}, a network is localizable if and only if $D_{ff}$ is nonsigular. We only need to prove that $D_{ff}$ is singular if and only if there exists a nonzero $\delta p \in \text{Null}(D)$ with $\delta p_a=\mathbf{0}$. (Sufficiency) If $\delta p \in \text{Null}(D)$ with $\delta p_a=\mathbf{0}$ and $\delta p_f \neq \mathbf{0}$, we have $\delta p_f^TD_{ff}\delta p_f=\delta p^TD\delta p=0$. Hence, $D_{ff}$ is singular. (Necessity) If $D_{ff}$ is singular, there exists a nonzero $\delta p_f$ satisfying $D_{ff}\delta p_f=\mathbf{0}$. Let $\delta p = (\delta p_a^T, \delta p_f^T)^T$ with $\delta p_a=\mathbf{0}$. We have $\delta p^TD\delta p=\delta p_f^TD_{ff}\delta p_f=0$. Hence, there exists a nonzero $\delta p \in \text{Null}(D)$ with $\delta p_a=\mathbf{0}$.
\end{proof}

Before exploring the topological condition for network localizability, we introduce two kinds of networks: non-coplanar network and coplanar network.  A network ($\mathcal{G}$, $p$) is called a non-coplanar network ($n > 4$) in $\mathbb{R}^3$ if its nodes are non-coplanar. A network ($\mathcal{G}$, $p$) is called a coplanar network ($n > 3$) if its nodes are coplanar but non-colinear.

\begin{lemma}\label{anchorn}

If there are only four coplanar anchor nodes $(n_a\!=\!4)$ in a non-coplanar network ($\mathcal{G}$, $p$), ($\mathcal{G}$, $p$) must not be localizable.
\end{lemma}

The proof of \textit{Lemma \ref{anchorn}} is give in the Appendix-\textcolor{blue}{A}. Based on \textit{Lemma \ref{anchorn}}, it is easy to prove that there must be at least four non-coplanar anchor nodes $(n_a\! \ge \!4)$ if a non-coplanar network ($\mathcal{G}$, $p$) is localizable.

\begin{theorem}\label{t3}
(\textbf{Topological Condition for Localizability of Non-coplanar Network}) {Under Assumption \ref{ass3},}
a non-coplanar network ($\mathcal{G}$, $p$) with at least four non-coplanar anchor nodes $n_a \! \ge \! 4$ in $\mathbb{R}^3$ is localizable if it is infinitesimally angle-displacement rigid.
\end{theorem}
\begin{proof}
Since ($\mathcal{G}$, $p$) is infinitesimally angle-displacement rigid, we have $\text{Null}(D) = \text{Span}\{ \mathbf{1}_n\otimes I_3, p, (I_n \otimes A)p, A+A^T =\mathbf{0},  A \in \mathbb{R}^{3 \times 3} \}$. 
Hence, any angle-displacement infinitesimal motion $\delta p$ can be expressed as 
$\delta p= sp + 
r(I_n \otimes  {A})p +   \mathbf{1}_n\otimes \mathbf{t}, s, r \in \mathbb{R}, \mathbf{t} \in \mathbb{R}^3 $. Since there are at least four non-coplanar anchor nodes and no two anchor nodes collocate, there does not exist a nonzero $\delta p$ satisfying $\delta p_a = \mathbf{0}$.  From \text{Theorem \ref{the4}}, we can know that ($\mathcal{G}$, $p$)  is localizable.
\end{proof}

\begin{remark}
If a network is not infinitesimally angle-displacement rigid, it may be localizable. For example, for the networks in Fig. \ref{fig51}$\textcolor{blue}{\text{c}}$ and Fig. \ref{fig51}$\textcolor{blue}{\text{d}}$, they are localizable but not infinitesimally angle-displacement rigid because
there is not enough number of edges among the anchor nodes that can be used for constructing the angle constraints.  The networks in Fig. \ref{fig51}$\textcolor{blue}{\text{c}}$ and \ref{fig51}$\textcolor{blue}{\text{d}}$ have the same angle-displacement function, but they are not infinitesimally angle-displacement rigid
because
there is a nontrivial infinitesimal angle-displacement motion $\delta p$ moving the network in Fig. \ref{fig51}$\textcolor{blue}{\text{c}}$ to the network in Fig. \ref{fig51}$\textcolor{blue}{\text{d}}$,   which is not in $\text{Span}\{ \mathbf{1}_n\otimes I_3, p, (I_n \otimes A)p, A+A^T =\mathbf{0},  A \in \mathbb{R}^{3 \times 3} \}$.
\end{remark}

\begin{figure}[t]
\centering
\includegraphics[width=1\linewidth]{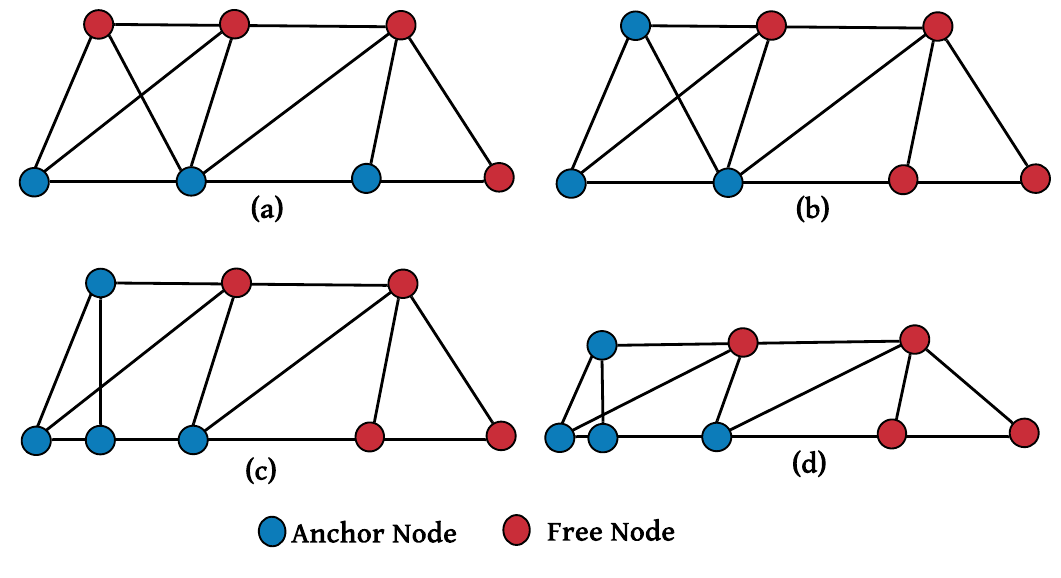}
\caption{Examples of coplanar networks. The blue dots are anchor nodes and the red dots are free nodes. The network in (a) is neither localizable nor infinitesimally angle-displacement rigid. The network in (b) is localizable and infinitesimally angle-displacement rigid.  The networks in (c) and (d) are localizable and have the same angle-displacement function, but they
are not infinitesimally angle-displacement rigid. }
\label{fig5}
\end{figure}

\subsection{Application to 2-D Network Localization}

\begin{assumption}\label{as31}
No two nodes are collocated. Each anchor node has at least two neighboring anchor nodes, and
each free node has at least three neighboring nodes. For the local-relative-bearing-based and angle-based networks, the free node and its neighbors are non-colinear. 
\end{assumption}

\begin{remark}
A node $\bar p_i=(x,y)^T$ in $\mathbb{R}^2$ can be represented by a node $p_i=(x,y,0)^T$ in $\mathbb{R}^3$.  Hence, 
2-D network can be regraded as a special coplanar network in $\mathbb{R}^3$. It is shown in Fig. \ref{moti}, if quadrilateral $\diamondsuit_{ijkh}(p)$ is obtained by rotating the quadrilateral $\diamondsuit_{ijkh}(q)$
about $Y$-axis by an angle $\pi$, 
quadrilaterals $\diamondsuit_{ijkh}(p)$ and $\diamondsuit_{ijkh}(q)$ will be on the same plane and have the same angle-displacement function $f_{\mathcal{T}_{\mathcal{G}}}(p)$. Note that this rotation can only be completed in $\mathbb{R}^3$ rather than in $\mathbb{R}^2$. Hence, the rigidity of the coplanar network should be studied in $\mathbb{R}^3$ rather than in $\mathbb{R}^2$.
\end{remark}

Next, we will show the relaxed displacement constraint in a coplanar network in $\mathbb{R}^3$.
In a coplanar network, for the free node $p_i$ and its neighbors $p_j,p_k,p_h$ in $\mathbb{R}^3$, we can know that there exists
a non-zero vector
$(\mu_{ij}, \mu_{ik}, \mu_{ih})^T \in \mathbb{R}^3$ (refer to Appendix-\textcolor{blue}{B}) such that
\begin{equation}\label{hy11}
    \mu_{ij}e_{ij}+\mu_{ik}e_{ik}+\mu_{ih}e_{ih}= \mathbf{0}.
\end{equation}

Using the methods presented in Section \ref{dissec} and \cite{fangtechnique}, 
the parameters $\mu_{ij}, \mu_{ik}, \mu_{ih}$ in \eqref{hy11} can also be obtained
by local relative position, distance, local relative bearing, angle, or ratio-of-distance measurements. Compared with the displacement constraint in \eqref{root},
the displacement constraint in \eqref{hy11} involves fewer nodes. Hence,
the corresponding assumption for constructing displacement constraint \eqref{hy11} can be relaxed as Assumption \ref{as31}.

\begin{lemma}\label{anchorn1}
If there are only three colinear anchor nodes $(n_a\!=\!3)$ in a coplanar network ($\mathcal{G}$, $p$), ($\mathcal{G}$, $p$) must not be localizable.
\end{lemma}

The proof of \textit{Lemma \ref{anchorn1}} is give in the Appendix-\textcolor{blue}{C}. 
Based on \textit{Lemma \ref{anchorn1}}, it is easy to prove that there must be at least three non-colinear anchor nodes $(n_a\!\ge\!3)$
if a coplanar network ($\mathcal{G}$, $p$) is localizable.

\begin{theorem}\label{cor2}
(\textbf{Topological Condition for Localizability of Coplanar Network}) {Under Assumption \ref{as31},}
a coplanar network ($\mathcal{G}$, $p$)  with at least three non-colinear anchor nodes $n_a \! \ge \! 3$ in $\mathbb{R}^3$ is localizable if it is infinitesimally angle-displacement rigid.
\end{theorem}

\begin{figure}[t]
\centering
\includegraphics[width=0.7\linewidth]{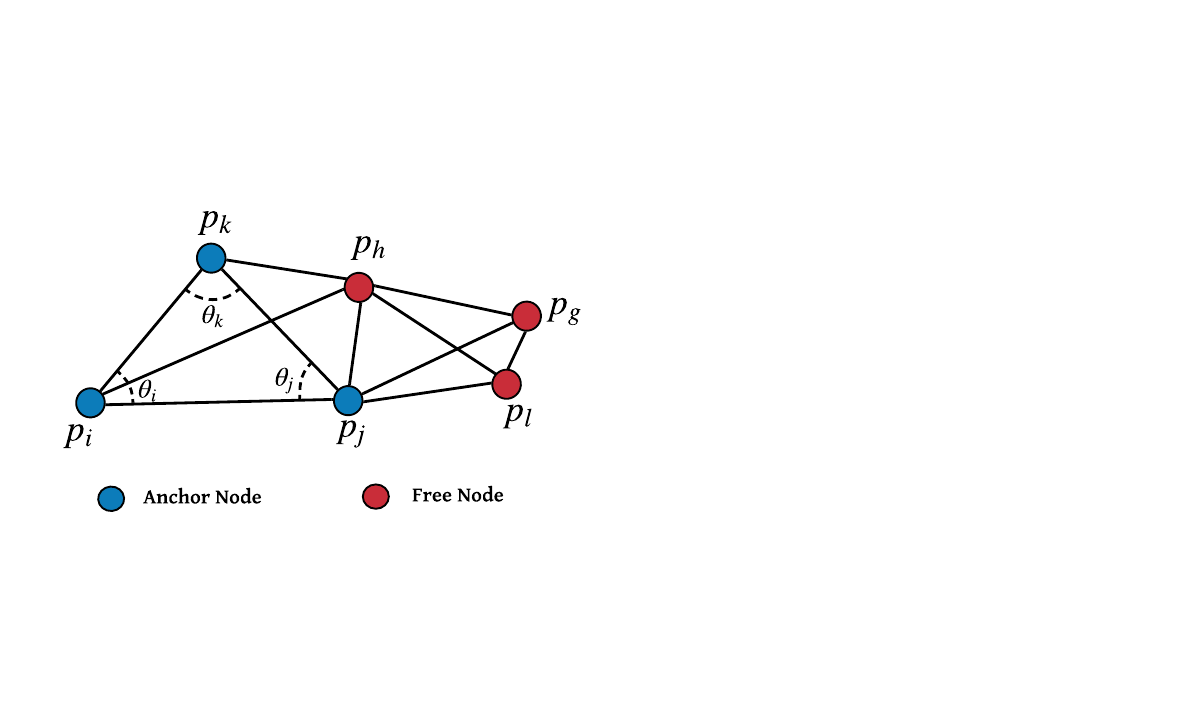}
\caption{Construction of 2-D localizable network. }
\label{motiaa}
\end{figure}

The proof of Theorem \ref{cor2} follows from that of \text{Theorem \ref{t3}}. 
Examples of coplanar networks are given in Fig. \ref{fig5}.
\text{Theorem \ref{cor2}} is applicable to 2-D network localization because  2-D network can be regraded as a coplanar network in $\mathbb{R}^3$. Since four nodes in 2-D or five nodes in 3-D can form a displacement constraint,
inspired by \cite{aspnes2006theory}, a simple sufficient graph condition is trilateration ordering for a 2-D localizable network or quadrilateral ordering for a 3-D localizable network.

Next,
an example of how to construct a 2-D localizable network is given in Fig. \ref{motiaa}, where the positions of anchor nodes are $p_i=(1,0)^T$, $p_j=(3,0)^T$, $p_k=(2,1)^T$.
\begin{enumerate}[(i)]
\item For the free node $p_h$,
the displacement constraint among $p_i,p_j,p_k,p_h$ is set as
$e_{ih}-3e_{jh}-2e_{kh}=0$, which can be rewritten as
\begin{equation}\label{lo1}
p_h=-\frac{1}{4}p_i+\frac{3}{4}p_j+\frac{1}{2}p_k.    
\end{equation}

\item For the free node $p_l$, the displacement constraint among $p_i,p_j,p_h,p_l$ is set as $5e_{il}-13e_{jl}-2e_{hl}=0$, which can be rewritten as
\begin{equation}
p_l=-\frac{1}{2}p_i+\frac{13}{10}p_j+\frac{1}{5}p_h.   
\end{equation}

\item For the free node $p_g$, the displacement constraint among $p_j,p_h,p_l,p_g$ is set as $e_{jg}-2e_{hg}-4e_{lg}=0$, which can be rewritten as
\begin{equation}\label{lo3}
p_g=-\frac{1}{5}p_j+\frac{2}{5}p_h+\frac{4}{5}p_l.     
\end{equation}

\end{enumerate}

Since the positions of anchor nodes $p_i,p_j,p_k$ are known, the positions of free nodes $p_h=(3,\frac{1}{2})^T,p_l=(4,\frac{1}{10})^T,p_g=(\frac{19}{5},\frac{7}{25})^T$ can be obtained by \eqref{lo1}-\eqref{lo3}.  Hence, the 2-D network in Fig. \ref{motiaa} is localizable.
Similarly, we can build a 3-D localizable network based on the known anchor nodes and displacement constraints.

\section{Distributed Network localization Protocol}\label{distri}

\subsection{Distributed Protocol}

Based on the cost function \eqref{cost1}, the  distributed localization algorithm is designed as
\begin{equation}\label{dis}
\dot  {\hat{p}}_f = -\bigtriangledown_{{\hat{p}}_f} ({\hat{p}}^T D {\hat{p}})= -D_{ff}{\hat{p}}_f - D_{fa}p_a.
\end{equation}

The algorithm \eqref{dis} is distributed because each free node $i \in  \mathcal{V}_f$ only uses
local information of $D$ 
that is related to it. 

\begin{remark}
The angle-displacement information matrix is obtained by $D=R(p)^TR(p)$, where
$R(p)$ can be obtained by the known anchor positions and local relative measurements (local relative position, distance, local relative bearing, angle, or ratio-of-distance measurements.)
\end{remark}

Based on \eqref{dis}, the
control protocol of node $i \in  \mathcal{V}_f$ is expressed as
\begin{equation}
\dot  {\hat{p}}_i  = \sum\limits_{(i,j,k,h,l) \in  \mathcal{X}_{\mathcal{G}}}M(i,j,k,h,l)+ \sum\limits_{(j,i,k,h,l) \in  \mathcal{X}_{\mathcal{G}}} \bar M(j,i,k,h,l), 
\end{equation}
where
\begin{equation}
\begin{array}{ll}
     & M(i,j,k,h,l) = (\mu_{ij}+\mu_{ik}+\mu_{ih}+\mu_{il}) \cdot  [\mu_{ij}(\hat{p}_j - \hat{p}_i)\!+\! \\
     &\mu_{ik}(\hat{p}_k - \hat{p}_i)\!+\!\mu_{ih}(\hat{p}_h - \hat{p}_i)\!+\!\mu_{il}(\hat{p}_l - \hat{p}_i)],
\end{array}
\end{equation}
and
\begin{equation}
    \begin{array}{ll}
     & \bar M(j,i,k,h,l) = \mu_{ji}^2(\hat{p}_j - \hat{p}_i)\!+\! \mu_{ji}\mu_{jk}(\hat{p}_j - \hat{p}_k)\!+\! \\
     & \mu_{ji}\mu_{jh}(\hat{p}_j - \hat{p}_h)\!+\!\mu_{ji}\mu_{jl}(\hat{p}_j - \hat{p}_l).
\end{array}
\end{equation}

\begin{theorem}\label{ts3}
(\textbf{Convergence of the Distributed Protocol}) {Under Assumption \ref{ass3},}
if a local-relative-measurement-based network ($\mathcal{G}$, $p$) with at least three non-colinear anchor nodes $n_a \! \ge \!  3$ is infinitesimally angle-displacement rigid in $\mathbb{R}^3$, given arbitrary initial estimates $\hat p_i(0), i \in \mathcal{V}_f$,  $\hat p_i(t)$ converges to $p_i$ globally and exponentially fast for all the free nodes $i \in \mathcal{V}_f$ under the distributed localization protocol \eqref{dis}.
\end{theorem}
\begin{proof}
From \text{Theorem \ref{the3}}, \text{Theorem} \ref{t3} and \textit{Remark \ref{rr2}},  we can know that the network ($\mathcal{G}$, $p$) is localizable and $D_{ff}$ 
is a positive definite matrix with $p_f = - D_{ff}^{-1}D_{fa}p_a$.
Consider a Lyapunov function $V = \frac{1}{2}\| {\hat{p}}_f - {{p}}_f \|^2$, we have 
\begin{equation}
\begin{array}{ll}
\dot V &= (\hat p_f - p_f)^T \dot  {\hat{p}}_f \\
& = (\hat p_f - p_f)^T(-D_{ff}{\hat{p}}_f - D_{fa}p_a) \\
& = (\hat p_f - p_f)^T(-D_{ff}{\hat{p}}_f + D_{ff}p_f) \\
& = -(\hat p_f - p_f)^TD_{ff}(\hat p_f - p_f) <0, \ \text{if} \ \hat p_f \neq p_f.
\end{array}
\end{equation}
Then, the conclusion follows.
\end{proof}

\begin{theorem}
(\textbf{Relationship between Distributed Protocol and Network Localizability})
The distributed localization protocol \eqref{dis} can globally localize the network ($\mathcal{G}$, $p$) if and only if the network is localizable.
\end{theorem}
\begin{proof}
(Sufficiency) If the network ($\mathcal{G}$, $p$) is localizable, 
$D_{ff}$ is nonsigular. From \text{Theorem \ref{ts3}}, we can know that the distributed localization protocol \eqref{dis} can globally localize the network. (Necessity) If the distributed localization protocol \eqref{dis} can globally localize the network, the true positions of free nodes $p_f$ is the unique solution to the cost function \eqref{pro}.
From \textit{Lemma \ref{ls4}}, we can know that $p_f$ is the unique solution to $D_{ff}\hat p_f +D_{fa}p_a= \mathbf{0}$. Hence, $D_{ff}$ must be nonsigular. From \text{Theorem \ref{the3}}, we can know that the network is localizable.
\end{proof}

\subsection{Position Estimation Error Under Noisy Measurements}

Considering the measurement noises in the position estimation, denote $\Delta  D_{ff}$ and $\Delta  D_{fa}$ as the error matrices due to the measurement noise and {write $\hat{D}_{ff} = D_{ff} \!+\! \Delta  D_{ff}$ and $\hat{D}_{fa} = D_{fa} \!+\! \Delta  D_{fa}$.}
The distributed algorithm \eqref{dis} becomes 
\begin{equation}
\dot  {\hat{p}}_f= -\hat{D}_{ff}{\hat{p}}_f - \hat{D}_{fa}p_a.
\end{equation}

If $\hat{D}_{ff}$ is nonsingular, the final estimate is given by
\begin{equation}
{\hat{p}}_f^* = -\hat{D}_{ff}^{-1}\hat{D}_{fa}p_a.
\end{equation}

\begin{theorem}
(\textbf{Nonsingular Condition of Matrix $D_{ff}$ with Noise}) Given a localizable network with $D_{ff}$ nonsingular, the matrix $\hat{D}_{ff}$ is nonsingular if the error matrix $\Delta  D_{ff}$ satisfies
\begin{equation}\label{perror}
\| \Delta  D_{ff} \| <  \lambda_{\min}(D_{ff}),
\end{equation}
where $\lambda_{\min}(D_{ff})$ is the smallest eigenvalue of $D_{ff}$. 
\begin{proof}
$D_{ff}$ is a symmetric positive definite matrix. 
Since $\| \Delta  D_{ff} \|  < \lambda_{\min}(D_{ff}) = \frac{1}{\| D_{ff}^{-1} \| }$, we have $\| D_{ff}^{-1} \Delta  D_{ff}   \| \le \| D_{ff}^{-1} \| \|  \Delta D_{ff} \| <1$. Hence,  the matrix $\hat{D}_{ff} = D_{ff} \!+\! \Delta  D_{ff} = D_{ff} (I \!+\! D_{ff}^{-1}\Delta  D_{ff})$ is nonsingular. 
	
\end{proof}
\end{theorem}

\begin{figure}[t]
\centering
\includegraphics[width=0.8\linewidth]{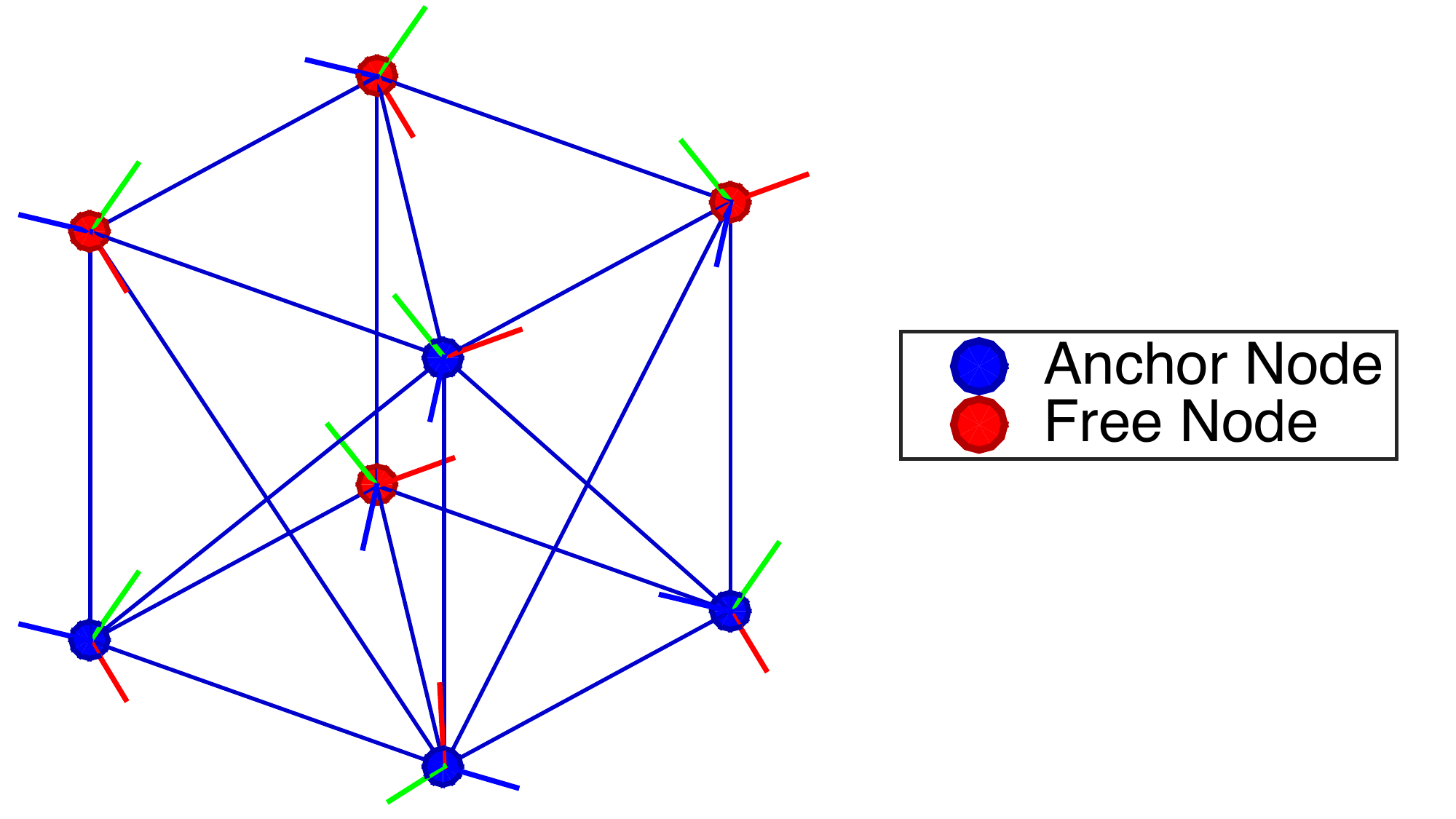}
\caption{{The 3-D network to be localized.} }
\label{case9}
\end{figure}

Since $p_f = -{D}_{ff}^{-1}{D}_{fa}p_a$ and $({D}_{ff}\!+\! \Delta {D}_{ff})^{-1}= {D}_{ff}^{-1}- {D}_{ff}^{-1}\Delta {D}_{ff}(I+{D}_{ff}^{-1}\Delta {D}_{ff})^{-1}{D}_{ff}^{-1}$ \cite{henderson1981deriving}, the position estimation error is 
\begin{equation}
\begin{array}{lll}
& \| {\hat{p}}_f^* - {{p}}_f\| \\
&= \| -({D}_{ff}\!+\! \Delta {D}_{ff})^{-1}({D}_{fa}\!+\! \Delta {D}_{fa})p_a\!+\! {D}_{ff}^{-1}{D}_{fa}p_a \|  \\
&= \| -(I \!+\! D_{ff}^{-1}\Delta D_{ff} )^{-1}D_{ff}^{-1} \Delta D_{fa} p_a  + \\
& \ \ \ \ D_{ff}^{-1}\Delta D_{ff}(I \!+\! D_{ff}^{-1}\Delta D_{ff} )^{-1}p_f  \| \\
& \le \| (I \!+\! D_{ff}^{-1}\Delta D_{ff} )^{-1}D_{ff}^{-1} \Delta D_{fa} p_a \| + \\
& \ \ \ \ \| D_{ff}^{-1}\Delta D_{ff}(I \!+\! D_{ff}^{-1}\Delta D_{ff} )^{-1}p_f\|  \\
& \le \frac{\|  \Delta D_{fa} \| \|p_a\| +\|  \Delta D_{ff} \| \|p_f\|}{\|I \!+\! D_{ff}^{-1}\Delta D_{ff}  \| \| D_{ff} \|}.
\end{array}
\end{equation}

\begin{figure*}[htbp]
\centering
\begin{subfigure}[t]{0.47\textwidth}
\centering
\includegraphics[width=1\linewidth]{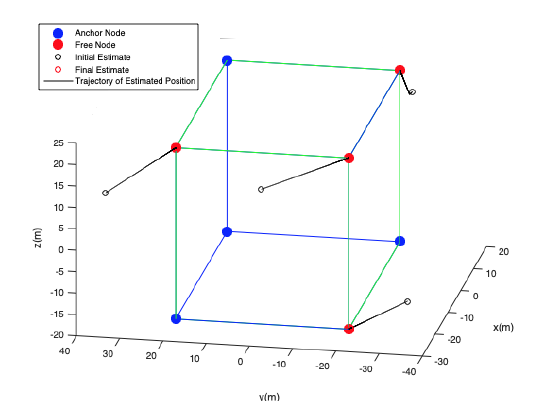}
\caption{Trajectories of estimated positions without measurement noise.}
\end{subfigure}
\begin{subfigure}[t]{0.47\textwidth}
\centering
\includegraphics[width=1\linewidth]{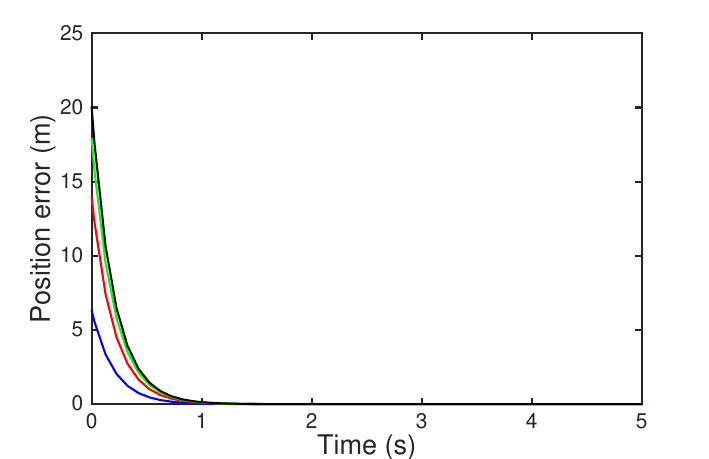}
\caption{Position errors without measurement noise.}
\end{subfigure}
\caption{Network localization without measurement noise.}
\label{case8}
\end{figure*}

\begin{figure*}[!t]
		\centering
		\begin{subfigure}[t]{0.45\textwidth}
		\centering
		\includegraphics[width=1\linewidth]{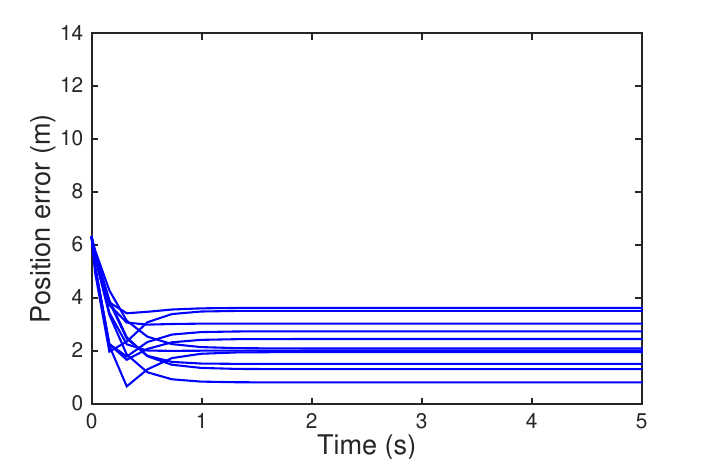}
		\caption{Position error of the free node $p_5$.}
		\label{e1}
	\end{subfigure}
	\begin{subfigure}[t]{0.45\textwidth}
		\centering
		\includegraphics[width=1\linewidth]{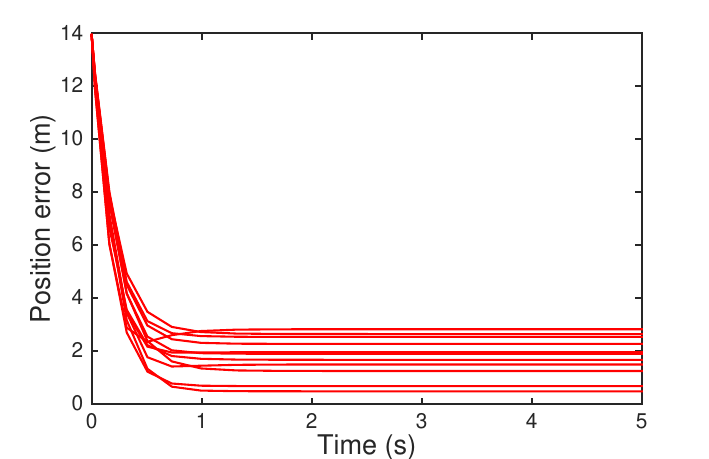}
		\caption{Position error of the free node $p_6$.}
		\label{e2}
	\end{subfigure}
	\begin{subfigure}[t]{0.45\textwidth}
	\centering
	\includegraphics[width=1\linewidth]{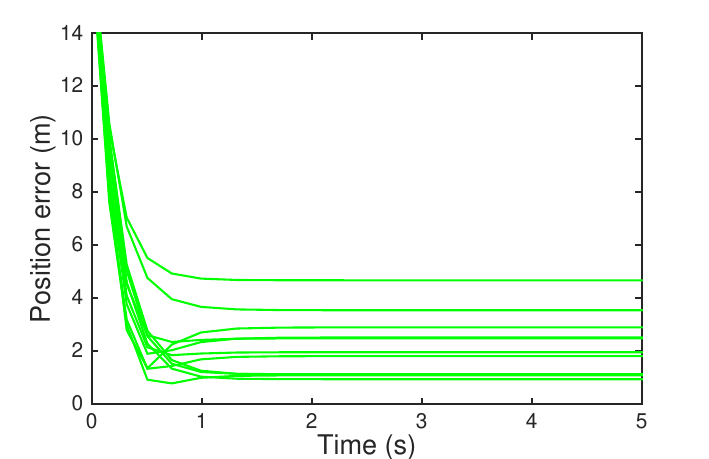}
     \caption{Position error of the free node $p_7$.}
	\label{e3}
	\end{subfigure}
    \begin{subfigure}[t]{0.45\textwidth}
	\centering
	\includegraphics[width=1\linewidth]{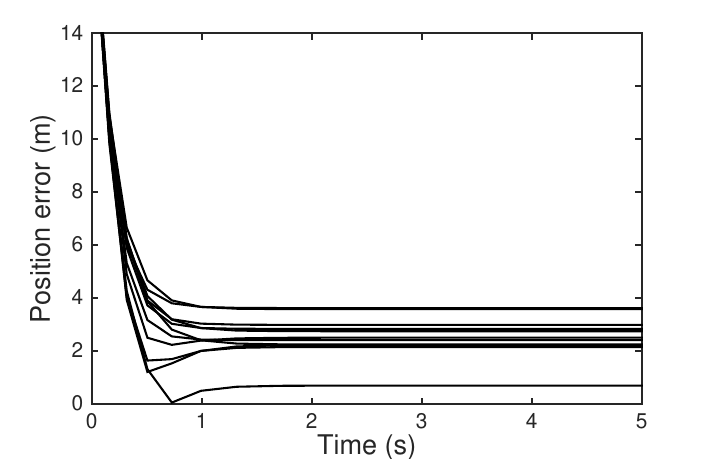}
     \caption{Position error of the free node $p_8$.}
	\label{e4}
	\end{subfigure}
 	 \caption{Network localization with measurement noise.}\label{e}
\end{figure*}

\subsection{Comparison with the Existing Distributed Network Localization Methods}

\begin{defn}\label{gener}
A configuration $p=(p_1^T, \cdots, p_{n}^T)^T$ is said to be generic if the coordinates $p_1, \cdots, p_n$ do not satisfy any nontrivial algebraic equation with integer coefficients. {Intuitively speaking, a generic configuration has no degeneracy, e.g., any three nodes are not on the same line in $\mathbb{R}^2$ and any four nodes are not on the same plane in $\mathbb{R}^3$ \cite{diao2014barycentric,han2017barycentric,lin2015distributed}.}
\end{defn}

The comparison with the existing distributed network localization methods is shown in Table \ref{tab:accuracy-orb}. The existing methods in relative-position-based,  distance-based, and bearing-based network localization require the generic configuration (refer to \text{Definition \ref{gener}}) or the known global coordinate frame.  

The generic configuration is a quite strict configuration, i.e., the parameters $\mu_{ij}, \mu_{ik}, \mu_{ih}, \mu_{il}$ in \eqref{root} and \eqref{hy11} in a generic configuration must be irrational numbers such as $\sqrt{2}$.  
Compared with the existing 2-D local-relative-position-based \cite{lin2015distributed},  2-D and 3-D distance-based  \cite{diao2014barycentric, han2017barycentric}  which require the generic configuration, our proposed method can be applied in both non-generic and generic configuration, i.e., the $\mu_{ij}, \mu_{ik}, \mu_{ih}, \mu_{il}$ in \eqref{root} and \eqref{hy11} can be either rational number or irrational number. 
In addition, compared with the existing 3-D global-relative-position-based \cite{barooah2007estimation} and 3-D global-relative-bearing-based network localization \cite{zhao2016localizability} which require the global coordinate frame, our proposed method can be applied in a network with unknown global coordinate frame. 

To our knowledge,
there exists no such known results for the localizability and distributed protocols of  angle-based networks in $\mathbb{R}^3$ and  ratio-of-distance-based networks in $\mathbb{R}^2$ and $\mathbb{R}^3$. Our proposed method solves angle-based distributed network localization in $\mathbb{R}^3$ and ratio-of-distance-based distributed network localization in $\mathbb{R}^2$ and $\mathbb{R}^3$.

\section{Simulation}\label{simulation}

A non-generic configuration with an unknown global coordinate frame in $\mathbb{R}^3$ is shown in Fig. \ref{case9}.
This network consists of four anchor nodes and four free nodes.
The positions of anchor nodes are $p_1\!=\!(20,-20,-20)^T$, $p_2\!=\!(20,20,-20)^T$, $p_3\!=\!(-20,20,-20)^T$ and $p_4\!=\!(20,20,20)^T$, respectively. The positions of free nodes are $p_5\!=\!(-20,-20,-20)^T$, $p_6\!=\!(20,-20,20)^T$, $p_7\!=\!(-20,20,20)^T$ and $p_8\!=\!(-20,-20,20)^T$, respectively.
{Since there are four coplanar nodes $p_4,p_6,p_7,p_8$ and their coordinates satisfies $p_4-p_6-p_7+p_8=0$, 
according to \text{Definition \ref{gener}},
the configuration $p=(p_1^T,\cdots,p_8^T)^T$ is non-generic. The red-green-blue lines represent the unknown local frame of each free node shown in Fig. \ref{case9}. 
The initial estimate of the position of each free node is randomly generated. The angle and displacement constraints are obtained by the known anchor positions and local relative measurements (local relative position, distance, local relative bearing, angle, or ratio-of-distance measurements).}
When there is no measurement noise,
it is shown in Fig. \ref{case8} that the position error of each free node converges to $0$. 
To simulate the noisy environment, 
each element of the error matrices $\Delta D_{ff}$ and $\Delta D_{fa}$ is generated 
by using a
white noise. 
Different measurement noises will result in different final estimates of the free nodes. We conduct $10$ simulations. The position errors of the free nodes are shown in Fig. \ref{e}(a)-(d).  The problem of how to evaluate the error matrices $\Delta D_{ff}, \Delta D_{fa}$ given the noise statistics of various types of local relative  measurements is beyond the scope of this paper, which will be considered in our future work.

\section{Conclusion}\label{coc}

This paper introduced a novel angle-displacement rigidity theory to investigate whether
a set of displacement constraints and angle constraints can uniquely characterize a network up to directly similar transformations, i.e., translation, rotation, and scaling. Based on the proposed
angle-displacement rigidity theory, a unified local-relative-measurement-based distributed network localization method is proposed, which can be used in local-relative-position-based, distance-based, local-relative-bearing-based,  angle-based, and ratio-of-distance-based distributed network localization. We provide necessary and sufficient conditions for network localizability. In addition, a distributed network localization protocol is proposed to globally estimate a network if the network is infinitesimally angle-displacement rigid.

The future work will be on network localization and formation control with mixed measurements based on the proposed angle-displacement rigidity theory.

\section*{Appendix}\label{app1}

\begin{figure}[t]
\centering
\includegraphics[width=0.9\linewidth]{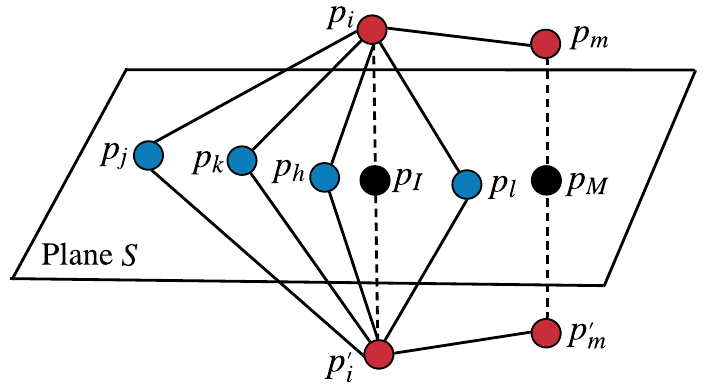}
\caption{An illustration of \textit{Lemma \ref{anchorn}}. }
\label{moti2}
\end{figure}

\subsection{Proof of \textit{Lemma} \ref{anchorn}}

For a non-coplanar network ($\mathcal{G}$, $p$), if there are only four anchor nodes $p_j, p_k, p_h, p_l$ which are on a plane $S$ in $\mathbb{R}^3$ shown in Fig. \ref{moti2}, $p_i, p_m$ are free nodes that are not on the plane $S$. $p'_i, p'_m$ are the symmetry nodes of $p_i, p_m$ with respect to the plane $S$.  
$p_I, p_M$ are the projections of $p_i, p_m$
on the plane $S$.
The displacement constraint \eqref{root}, based on different number of anchor nodes and free nodes, can be divided into five cases: 
(\text{\romannumeral1}) Four anchor nodes and one free node; (\text{\romannumeral2}) Three anchor nodes and two free nodes; (\text{\romannumeral3}) Two anchor nodes and three free nodes; (\text{\romannumeral4}) One anchor node and four free nodes; (\text{\romannumeral5}) Five free nodes. The above five cases are analyzed below.

Case (\text{\romannumeral1}):
Without loss of generality,
the displacement constraint $\mu_{ij}e_{ij}+\mu_{ik}e_{ik}+\mu_{ih}e_{ih} + \mu_{il}e_{il}= \mathbf{0}$, involving one free node $p_i$ and four anchor nodes $p_j, p_k, p_h, p_l$, can be rewritten as
\begin{equation}\label{pl1}
 \mu_{ij}(e_{iI}\!+\!e_{Ij})\!+\!\mu_{ik}(e_{iI}\!+\!e_{Ik})\!+\!\mu_{ih}(e_{iI}\!+\!e_{Ih})\! + \! \mu_{il}(e_{iI}\!+\!e_{Il})\!= \!\mathbf{0}. 
\end{equation}

From \eqref{pl1}, we have
\begin{equation}
 (\mu_{ij}\!+\!\mu_{ik}\!+\!\mu_{ih}\!+\!\mu_{il})e_{iI} \!+\!\mu_{ij}e_{Ij}\!+\!\mu_{ik}e_{Ik}\!+\!\mu_{ih}e_{Ih}\! + \! \mu_{il}e_{Il}\!= \!\mathbf{0}. 
\end{equation}

Since $e_{iI}$ is perpendicular to $e_{Ij}, e_{Ik}, e_{Ih}, e_{Il}$, we have $\mu_{ij}+\mu_{ik}+\mu_{ih}+\mu_{il}=0$. Then, the following equation must hold
\begin{equation}\label{pl}
 -(\mu_{ij}\!+\!\mu_{ik}\!+\!\mu_{ih}\!+\!\mu_{il})e_{iI} \!+\!\mu_{ij}e_{Ij}\!+\!\mu_{ik}e_{Ik}\!+\!\mu_{ih}e_{Ih}\! + \! \mu_{il}e_{Il}\!= \!\mathbf{0}. 
\end{equation}

Denote by $e'_{ij}=p_j-p'_i=e_{Ij}-e_{iI}, e'_{ik}=p_k-p'_i=e_{Ik}-e_{iI}, e'_{ih}=p_h-p'_i=e_{Ih}-e_{iI}, e'_{il}=p_l-p'_i=e_{Il}-e_{iI}$. From \eqref{pl}, we obtain
\begin{equation}\label{pl2}
    \mu_{ij}e'_{ij}+\mu_{ik}e'_{ik}+\mu_{ih}e'_{ih} + \mu_{il}e'_{il}= \mathbf{0}.
\end{equation}

Hence, $p_i, p_j, p_k, p_h, p_l$ and
$p'_i, p_j, p_k, p_h, p_l$ have the same displacement constraint. 

Case (\text{\romannumeral2}): Without loss of generality,
the displacement constraint $\bar \mu_{ij}e_{ij}+\bar \mu_{ik}e_{ik}+\bar \mu_{ih}e_{ih} + \bar \mu_{im}e_{im}= \mathbf{0}$,
involving two free node $p_i, p_m$ and three anchor nodes $p_j, p_k, p_h$, can be rewritten as 
\begin{equation}\label{pl11}
\begin{array}{ll}
     &  \bar \mu_{ij}(e_{iI}\!+\!e_{Ij})\!+\!\bar \mu_{ik}(e_{iI}\!+\!e_{Ik})\!+\! \bar \mu_{ih}(e_{iI}\!+\!e_{Ih})\! + \! \\
     & \bar \mu_{im}(e_{iI}\!+\!e_{IM}\!+\!e_{Mm})\!= \!\mathbf{0}.
\end{array}
\end{equation}

Since $e_{Mm}$ is parallel to $e_{iI}$, we have $e_{Mm}= \lambda e_{iI}, \lambda \in \mathbb{R}$.
\eqref{pl11} becomes
\begin{equation}
\begin{array}{ll}
& (\bar \mu_{ij}\!+\!\bar \mu_{ik}\!+\!\bar \mu_{ih}\!+\! (1+\lambda) \bar \mu_{im})e_{iI} \!+\!\bar \mu_{ij}e_{Ij}\!+\!\bar \mu_{ik}e_{Ik}\!+\!\\
& \bar \mu_{ih}e_{Ih}\! + \! \bar \mu_{im}e_{IM}\!= \!\mathbf{0}. 
\end{array}
\end{equation}
 
Since $e_{iI}$ is perpendicular to $e_{Ij}, e_{Ik}, e_{Ih}, e_{IM}$, we have $\bar \mu_{ij}+\bar \mu_{ik}+\bar \mu_{ih}+(1+\lambda)\bar \mu_{im}=0$. Then, the following equation must hold
\begin{equation}\label{plle}
\begin{array}{ll}
& -(\bar \mu_{ij}\!+\!\bar \mu_{ik}\!+\!\bar \mu_{ih}\!+\! (1+\lambda) \bar \mu_{im})e_{iI} \!+\!\bar \mu_{ij}e_{Ij}\!+\!\bar \mu_{ik}e_{Ik}\!+\!\\
& \bar \mu_{ih}e_{Ih}\! + \! \bar \mu_{im}e_{IM}\!= \!\mathbf{0}. 
\end{array}
\end{equation}

Denote by $e'_{ij}=p_j-p'_i=e_{Ij}-e_{iI}, e'_{ik}=p_k-p'_i=e_{Ik}-e_{iI}, e'_{ih}=p_h-p'_i=e_{Ih}-e_{iI}, e'_{im}=p'_m-p'_i=-e_{iI}+e_{IM}-e_{Mm}$. From \eqref{plle}, we obtain
\begin{equation}\label{pll2}
    \bar \mu_{ij}e'_{ij}+\bar \mu_{ik}e'_{ik}+\bar \mu_{ih}e'_{ih} + \bar \mu_{im}e'_{im}= \mathbf{0}.
\end{equation}

Hence, $p_i, p_j, p_k, p_h, p_m$ and
$p'_i, p_j, p_k, p_h, p'_m$ have the same displacement constraint. 
Similarly, for the cases (\text{\romannumeral3)-(\romannumeral5}), we can also prove that a displacement constraint \eqref{root} will remain unchanged by replacing the free nodes with the symmetry nodes of the free nodes with respect to the plane $S$.  

Denote $p_a$ and $p_f$ as the anchor nodes set and free nodes set, respectively. Denote $p'_f$ as the
symmetry nodes of the free nodes $p_f$ with respect to the plane $S$. Denote by $(\mathcal{G}, p)$ with $p=(p_a^T,p_f^T)^T$.  Denote by $(\mathcal{G}, p')$ with $p'=(p_a^T,{p'_f}^T)^T$. Since $(\mathcal{G}, p)$ and $(\mathcal{G}, p')$ have the same anchor nodes, i.e., the angle constraints among the anchor nodes are the same. From the cases (\text{\romannumeral1)-(\romannumeral5}) analyzed above, $(\mathcal{G}, p)$ and $(\mathcal{G}, p')$ have the same displacement constraints. Then, we can know that  $(\mathcal{G}, p)$ and $(\mathcal{G}, p')$ 
have the same angle-displacement function. Hence, 
there is an angle-displacement infinitesimal motion moving from $(\mathcal{G}, p)$ to $(\mathcal{G}, p')$ that involves no anchor node. From \text{Theorem \ref{the4}},  we can know that 
($\mathcal{G}$, $p$) is not localizable. Similarly, 
we can prove that a non-coplanar network ($\mathcal{G}$, $r$) is not localizable if there are less than four anchor nodes.

\subsection{Displacement Constraint in a Plane}

\begin{defn}\label{df1}
A plane in 3-D space is described with a single linear equation of the following form
\begin{equation}\label{hy}
  \varphi^Tb= c, 
\end{equation}
where $b \in \mathbb{R}^3$ is a non-zero vector and $c \in \mathbb{R}$ is a constant. $ \varphi \in \mathbb{R}^3$ is any point on the plane \eqref{hy}.
\end{defn}

If $p_i, p_j, p_k, p_h$ are on a plane $\varphi^Tb = c$ in 3-D space, we have
\begin{equation}\label{pq1}
\begin{array}{cccc}
&  p_{i}^T b = c, \ \ p_{j}^T b = c, \ \ p_{k}^T b = c, \ \ p_{h}^T b = c. 
\end{array}
\end{equation}

Equation \eqref{pq1} can be rewritten as
\begin{equation}
\begin{array}{c}
E_i^T b = \mathbf{0},
\end{array}
\end{equation}
where $E_i = (e_{ij}, e_{ik}, e_{ih}) \in \mathbb{R}^{3 \times 3}$. From \text{Definition \ref{df1}}, we can know that
$b$ is a non-zero vector. Hence, the matrix $E_i^T$ is not full rank. Since $\text{Rank}(E_i)= \text{Rank}(E_i^T)$, the matrix $E_i$ is also not full rank. From the matrix theory, there exists a non-zero vector $\mu_i=(\mu_{ij}, \mu_{ik}, \mu_{ih})^T \in \mathbb{R}^3$ such that
\begin{equation}\label{pq2}
 E_i \mu_i = \mu_{ij}e_{ij}+
 \mu_{ik}e_{ik}+\mu_{ih}e_{ih}.
\end{equation}

\subsection{Proof of \textit{Lemma} \ref{anchorn1}}

For a coplanar network ($\mathcal{G}$, $p$) in $\mathbb{R}^3$,  if there are only three anchor nodes $p_j, p_k, p_h$ that are on a line $L$ shown in Fig. \ref{moti3}, $p_i, p_m$ are free nodes that are not on the line $L$. 
$p_I, p_M$ are the projections of $p_i, p_m$
on the line $L$, and $p'_i, p'_m$ are the symmetry nodes of 
$p_i, p_m$ with respect to the line $L$ on the plane $S$. 
The displacement constraint \eqref{hy11}, based on different number of anchor nodes and free nodes, can be divided into four cases: 
(\text{\romannumeral1}) Three anchor nodes and one free node; (\text{\romannumeral2}) Two anchor nodes and two free nodes; (\text{\romannumeral3}) One anchor node and three free nodes; (\text{\romannumeral4}) Four free nodes. The above four cases are analyzed below.

Case (\text{\romannumeral1}): Without loss of generality, the 
displacement constraint $\mu_{ij}e_{ij}+\mu_{ik}e_{ik}+\mu_{ih}e_{ih} = \mathbf{0}$, involving one free node $p_i$ and three anchor nodes $p_j, p_k, p_h$, can be rewritten as
\begin{equation}\label{pll1}
 \mu_{ij}(e_{iI}\!+\!e_{Ij})\!+\!\mu_{ik}(e_{iI}\!+\!e_{Ik})\!+\!\mu_{ih}(e_{iI}\!+\!e_{Ih})\!= \!\mathbf{0}. 
\end{equation}

From \eqref{pll1}, we have
\begin{equation}
 (\mu_{ij}\!+\!\mu_{ik}\!+\!\mu_{ih})e_{iI} \!+\!\mu_{ij}e_{Ij}\!+\!\mu_{ik}e_{Ik}\!+\!\mu_{ih}e_{Ih}\!= \!\mathbf{0}. 
\end{equation}

Since $e_{iI}$ is perpendicular to $e_{Ij}, e_{Ik}, e_{Ih}$, we have $\mu_{ij}+\mu_{ik}+\mu_{ih}=0$. Then, the following equation must hold
\begin{equation}\label{pll}
 -(\mu_{ij}\!+\!\mu_{ik}\!+\!\mu_{ih})e_{iI} \!+\!\mu_{ij}e_{Ij}\!+\!\mu_{ik}e_{Ik}\!+\!\mu_{ih}e_{Ih}\!= \!\mathbf{0}. 
\end{equation}

\begin{figure}[t]
\centering
\includegraphics[width=0.8\linewidth]{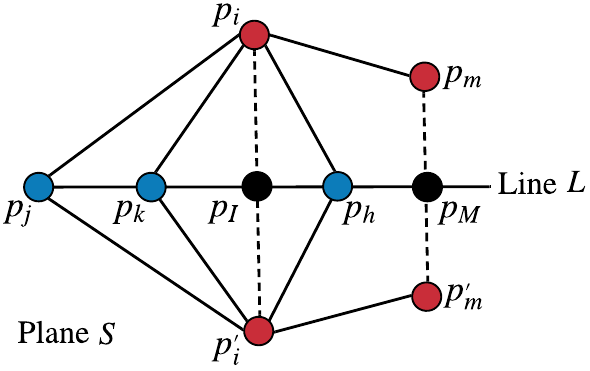}
\caption{An illustration of \textit{Lemma \ref{anchorn1}}. }
\label{moti3}
\end{figure}

Denote by  $e'_{ij}=p_j-p'_i=e_{Ij}-e_{iI}, e'_{ik}=p_k-p'_i=e_{Ik}-e_{iI}, e'_{ih}=p_h-p'_i=e_{Ih}-e_{iI}$. 
From \eqref{pll}, we obtain
\begin{equation}
    \mu_{ij}e'_{ij}+\mu_{ik}e'_{ik}+\mu_{ih}e'_{ih} = \mathbf{0}.
\end{equation}

Hence, $p_i, p_j, p_k, p_h$ and
$p'_i, p_j, p_k, p_h$ have the same displacement constraint.

Case (\text{\romannumeral2}): Without loss of generality,
the displacement constraint $\bar \mu_{ij}e_{ij}+\bar \mu_{ik}e_{ik} + \bar \mu_{im}e_{im}= \mathbf{0}$,
involving two free node $p_i, p_m$ and two anchor nodes $p_j, p_k$, can be rewritten as 
\begin{equation}\label{ppl11}
      \bar \mu_{ij}(e_{iI}\!+\!e_{Ij})\!+\!\bar \mu_{ik}(e_{iI}\!+\!e_{Ik})\!+\!  \bar \mu_{im}(e_{iI}\!+\!e_{IM}\!+\!e_{Mm})\!= \!\mathbf{0}.
\end{equation}

Since $e_{Mm}$ is parallel to $e_{iI}$, we have $e_{Mm}= \lambda e_{iI}, \lambda \in \mathbb{R}$.
\eqref{ppl11} becomes
\begin{equation}
(\bar \mu_{ij}\!+\!\bar \mu_{ik}\!+\! (1\!+\!\lambda) \bar \mu_{im})e_{iI} \!+\!\bar \mu_{ij}e_{Ij}\!+\!\bar \mu_{ik}e_{Ik}\!+\! \bar \mu_{im}e_{IM}\!= \!\mathbf{0}. 
\end{equation}

Since $e_{iI}$ is perpendicular to $e_{Ij}, e_{Ik}, e_{IM}$, we have $\bar \mu_{ij}+\bar \mu_{ik}+(1+\lambda)\bar \mu_{im}=0$. Then, the following equation must hold
\begin{equation}\label{pplle}
-(\bar \mu_{ij}\!+\!\bar \mu_{ik}\!+\! (1\!+\!\lambda) \bar \mu_{im})e_{iI} \!+\!\bar \mu_{ij}e_{Ij}\!+\!\bar \mu_{ik}e_{Ik}\!+\! \bar \mu_{im}e_{IM}\!= \!\mathbf{0}. 
\end{equation}

Denote by $e'_{ij}=p_j-p'_i=e_{Ij}-e_{iI}, e'_{ik}=p_k-p'_i=e_{Ik}-e_{iI}, e'_{im}=p'_m-p'_i=e_{IM}+e_{Mm}+ve_{iI}=-e_{iI}+e_{IM}-e_{Mm}$. 
From \eqref{pll}, we obtain
\begin{equation}
    \mu_{ij}e'_{ij}+\mu_{ik}e'_{ik}+\mu_{im}e'_{im} = \mathbf{0}.
\end{equation}

Hence, $p_i, p_j, p_k, p_m$ and
$p'_i, p_j, p_k,  p'_m$ have the same displacement constraint. 
Similarly, for the cases (\text{\romannumeral3)-(\romannumeral4}), we can also prove that a displacement constraint \eqref{hy11} will remain unchanged by replacing the free nodes with the symmetry nodes of the free nodes with respect to the line $L$ on the plane $S$.  

Denote $p_a$ and $p_f$ as the anchor nodes set and free nodes set, respectively. Denote  $p'_f$ as the
symmetry nodes of the free nodes $p_f$ with respect to the line $L$ on the plane $S$. Denote by $(\mathcal{G}, p)$ with $p=(p_a^T,p_f^T)^T$.  Denote by $(\mathcal{G}, p')$ with $p'=(p_a^T,{p'_f}^T)^T$. Since $(\mathcal{G}, p)$ and $(\mathcal{G}, p')$ have the same anchor nodes, i.e., the angle constraints among the anchor nodes are the same. From the cases (\text{\romannumeral1)-(\romannumeral4}) analyzed above, $(\mathcal{G}, p)$ and $(\mathcal{G}, p')$ have the same displacement constraints. Then, we can know that  $(\mathcal{G}, p)$ and $(\mathcal{G}, p')$ 
have the same angle-displacement function. Hence, 
there is an angle-displacement infinitesimal motion moving from $(\mathcal{G}, p)$ to $(\mathcal{G}, p')$ that involves no anchor node. From \text{Theorem \ref{the4}}, we can know that 
($\mathcal{G}$, $p$) is not localizable. Similarly, 
we can prove that a coplanar network ($\mathcal{G}$, $r$) is not localizable if there are less than three anchor nodes.

\bibliographystyle{IEEEtran}
\bibliography{papers}

\end{document}